\documentclass{LMCS}

\def\doi{8(4:12)2012}
\lmcsheading%
{\doi}
{1--41}
{}
{}
{Apr.~\phantom06, 2011}
{Oct.~23, 2012}
{}

\usepackage{amsmath, amsthm,latexsym,amssymb,xspace}
\usepackage{color}
\usepackage{xspace}
\usepackage{graphicx}
\usepackage{tikz}
\usepackage{enumerate,hyperref}

\usetikzlibrary{arrows}
\usetikzlibrary{trees}

\newcommand{\frM}{\mathfrak{M}}
\newcommand{\frN}{\mathfrak{N}}
\newcommand{\xml}{\textsf{XML}\xspace}

\newcommand{\fo}{\textsf{FO}\xspace}

\newcommand{\mso}{\textsf{MSO}\xspace}

\newcommand{\fodtc}{\textsf{FO(DTC$^1$)}\xspace}
\newcommand{\fotc}{\textsf{FO(TC$^1$)}\xspace}

\newcommand{\folfp}{\textsf{FO(LFP$^1$)}\xspace}

\newcommand{\foatc}{\textsf{FO(ATC$^1$)}\xspace}
\newcommand{\lfp}{\textsf{FO(LFP)}\xspace}

\newcommand{\tc}{\textrm{\sf TC}\xspace}

\newcommand{\finremark}{\hfill$\dashv$}

\newenvironment{proofofclaim}{\begin{trivlist}\item\textit{Proof of claim.}}{\hfill$\dashv$\end{trivlist}}
\input xy
\xyoption{all}

\tikzstyle{nodo}=[very thick,anchor=north]

\tikzstyle{noeud}=[anchor=north]

\begin{document}

\title[Complete Axiomatizations of Fragments of MSO on Finite Trees]{Complete Axiomatizations of Fragments of Monadic Second-Order Logic on Finite Trees}

\thanks{{\lsuper{a,b}}We are grateful to Jouko V\"a\"an\"anen for helpful comments on an earlier draft. Most of the work was done when the first author was at the ILLC, University of Amsterdam, supported by a GLoRiClass fellowship of the European Commission (Research Training Fellowship MEST-CT-2005-020841) and when the second author was at the ISLA, University of Amsterdam, supported by the Netherlands Organization
for Scientific Research (NWO) grant 639.021.508. The first author also
acknowledges the EPSRC grant EP/G049165/1 and the FET-Open Project FoX, grant agreement
233599. The second author
acknowledges the NSF grant IIS-0905276.}

\author[A.~Gheerbrant]{Am\'{e}lie Gheerbrant\rsuper a}
\address{{\lsuper a}School of Informatics, University of Edinburgh}
\email{agheerbr@inf.ed.ac.uk}

\author[B.~ten Cate]{Balder ten Cate\rsuper b}
\address{{\lsuper b}Department of Computer Science, University of
  California, Santa Cruz}
\email{btencate@ucsc.edu}

\begin{abstract}
We consider a specific class of tree structures that can 
represent basic structures in linguistics and computer science such as \xml
documents, parse trees, and treebanks, namely, finite node-labeled sibling-ordered
trees. We present axiomatizations of the monadic second-order logic (\mso),
monadic transitive closure logic (\fotc) and monadic least fixed-point logic (\folfp) theories of this class of structures. 
These logics can express important properties such as reachability. 
Using model-theoretic techniques, we show by a uniform argument that these axiomatizations are
complete, i.e., each formula that is valid on all finite trees is provable using our axioms.
As a backdrop to our positive results, on arbitrary structures, the logics that we study are known to be
non-recursively axiomatizable.
\end{abstract}


\maketitle

\section{Introduction}

In this paper, we develop a uniform method for obtaining
complete axiomatizations of fragments of \mso on finite trees. In
particular, we obtain a complete axiomatization for \mso, \fotc,
and \folfp on finite node-labeled sibling-ordered trees. We take
inspiration from Kees Doets, who proposed in \cite{1987}
complete axiomatizations of \fo-theories in particular on the class of
node-labeled finite trees without sibling-order  
(see Section \ref{sec3}, where we discuss his work in more details). A similar result
for \fo on node-labeled finite trees with sibling order was
shown in \cite{1995} in the context of model-theoretic syntax
and in \cite{07} in the context of \xml query languages. We use the signature of \cite{07}
and extend the set of axioms proposed there to match the richer syntax of the logics we consider.

Finite trees are basic and ubiquitous structures that are of
interest at least to mathematicians, computer scientists
(e.g. tree-structured documents) and linguists (e.g. parse trees). The
logics we study are known to be very well-behaved on this
particular class of structures and to have an interestingly high
expressive power. In particular, they all allow to express
reachability, but at the same time, they have the advantage of
being decidable on trees.

As \xml documents are tree-structured data, our results are
relevant to \xml query languages. Declarative query languages fro both relational and \xml data are
based on logical languages. In \cite{1376952} and \cite{GottlobKoch},
\mso and \fotc have been proposed as a yardstick of expressivity
of navigational query languages for \xml. It is known that \folfp has the same
expressive power as \mso on trees, but the translations between
the two are non-trivial, and hence it is not clear whether an
axiomatization for one language can be obtained from an
axiomatization for the other language in any straightforward way.
One important and well-studied problem for XML query languages, as well as for
database query languages in general,  is
\emph{query optimization}. Typically, a query can be expressed in many
equivalent ways, and the execution time of a query depends strongly on
the way it is expressed. A common approach to database query optimization is
by means of a set of rewrite rules, allowing one to
transform a query expression into another equivalent one, together with
a cost model that predicts the execution time of a query expression on
a given database \cite{DBLP:books/aw/AbiteboulHV95}.
 In~\cite{CateMarxXPath20}, a sound and complete set of
rewrite rules for the XML path language Core XPath 2.0 was obtained
from a complete axiomatization of the first-order theory of finite trees,
exploiting the fact that Core XPath 2.0 is expressively complete for
first-order logic. We expect that the results we present here can be
used in order to obtain sound and complete sets of rewrite rules for
dialects of Core XPath that are expressively complete for \fotc and for
\mso, such as the ones presented in \cite{TenCate2006,1376952}.

In applications to computational linguistics, finite trees are
used to represent the grammatical structure of natural language
sentences. In the context of \emph{model theoretic syntax}, Rogers
advocates in \cite{521965} the use of \mso in order to
characterize derivation trees of context free grammars. Kepser
also argues in \cite{1047003} that \mso should be used
in order to query treebanks. A treebank is a text corpus in which
each sentence has been annotated with its syntactic structure
(represented as a tree structure). In \cite{2006} and \cite{1219706} Kepser and
Tiede propose to consider various transitive closure logics,
among which \fotc, arguing that they constitute very natural
formalisms from the logical point of view, allowing concise and intuitive phrasing of parse tree properties.\\

The remainder of the paper is organized as follows: in Section \ref{sec1} we present the concept of finite tree and the logics we are interested in together with their standard interpretation. 
Section \ref{sec2} merely state our three
axiomatizations. In Section \ref{sec3}, we introduce non standard semantics called \textit{Henkin
semantics} for which our axiomatizations are easily seen to be
complete. We prove in detail the \folfp Henkin completeness proof. Section \ref{sec4} introduces operations on Henkin structures: substructure formation and a general operation of Henkin structures combination. We obtain Feferman-Vaught theorems for this operation by means of Ehrenfeucht-Fra\"{\i}ss\'e games. This section contains in particular the definitions and adequacy proofs of the Ehrenfeucht-Fra\"{\i}ss\'e games that we also use there to prove our Feferman-Vaught theorems.
In Section \ref{sec5}, we prove \textit{real} completeness (that is, on
the more restricted class of finite trees).
For that purpose, we consider substructures of trees
that we call forests and use the general operation discussed in Section \ref{sec4}
to combine a set of forests into one new forest. Our Feferman-Vaught theorems apply to such
constructions and we use them in our main proof of completeness, showing
that no formula of our language can distinguish Henkin models of
our axioms from real finite trees.
We also point out that every standard model of our axioms actually is a finite tree. 
Finally, we notice in Section \ref{secisftseclo} that a simplified version of our method can be used to show similar results for the class of node-labeled finite linear orders.

\section{Preliminaries}
\label{sec1}

\subsection{Finite Trees}

A tree is a partially ordered set with a unique element called \emph{the root} and such that apart from the root, each element (or \emph{node}) has one unique immediate predecessor. We are interested in
\emph{finite node-labeled sibling-ordered trees}: finite trees in which the
children of each node are linearly ordered. Also, the nodes can be
labeled by unary predicates. We will call these structures
\emph{finite trees} for short.
\begin{defi}[Finite tree]
  Let $\{P_1, \ldots,
  P_n\}$ be a fixed finite set of unary predicate symbols.  By a finite tree, we mean a finite structure
  $\frM=(M,<,\prec,P_1, \ldots, P_n)$, where $(M,<)$ is a tree (with
  $<$ the descendant relation) and $\prec$ linearly orders the
  children of each node.
\end{defi}

\subsection{Three Extensions of First-Order Logic}
\label{seclog}
In this section, we introduce three extensions of \fo : \mso, \fotc and \folfp. In the remainder of the paper (unless explicitly
stated otherwise), we will always be working with a fixed purely relational vocabulary $\sigma$ (i.e. with no individual constant or function symbols) and hence, with $\sigma$-structures. We assume as usual that we have a countably infinite set of first-order variables. In the case of
\mso and \folfp, we also assume that we have a countably infinite set of set variables. The semantics defined in this section we will refer to as
\textit{standard semantics} and the associated structures, as \textit{standard structures}.

We first introduce monadic second order logic, \mso, which is the
extension of first-order logic in which we can quantify over the subsets of the
domain.
\begin{defi}
[Syntax and semantics of \mso]
    Let $At$ stand for a first-order atomic formula (of the form
    $R(\vec{x}))$, $x=y$, or $\top$), $x$ a first-order variable and $X$ a set
    variable. The set of \mso formulas is given by the following
    recursive definition:
$$\varphi:= At ~|~ Xx ~|~ \varphi \wedge \psi ~|~ \varphi \vee \psi ~|~ \varphi
\rightarrow \psi ~|~ \neg \varphi ~|~ \exists x~ \varphi ~|~ \exists X~
\varphi ~$$

We use $\forall X \varphi$ (resp. $\forall x \varphi$) as shorthand for
$\neg \exists X \neg \varphi$ (resp. $\neg \exists x \neg \varphi$). We
define the \textit{quantifier depth} of a \mso formula as the maximal
number of first-order and second-order nested quantifiers.
We interpret \mso formulas in first-order structures.
Like for \fo formulas, the truth of \mso
formulas in $\frM$ is defined modulo a valuation $g$ of variables
as objects. But here, we also have set variables,
to which $g$ assigns
subsets of the domain. We let $g[a/x]$ be the
assignment which differs from $g$ only in assigning $a$ to $x$ (similarly for $g[A/X]$).
The truth of atomic formulas is defined by the usual
\fo clauses plus the following:
$$\frM, g \models Xx\text{ iff }g(x) \in g(X)\text{ for }X\text{ a set variable}$$

The truth of compound formulas is defined by induction, with the
same clauses as in \fo and an additional one:
\begin{center}
$\frM, g \models \exists X \varphi$ iff there is $A \subseteq M$ such
that $\frM, g[A/X] \models \varphi$
\end{center}
\end{defi}

The second logic we are interested in is monadic transitive
closure logic, \fotc, which extends \fo by closing it under the
transitive closure of binary definable relations.

\begin{defi}
[Syntax and semantics of \fotc]
    Let $At$ stand for a first-order atomic formula  (of the form
    $R(\vec{x}))$, $x=y$, or $\top$), $u,v,x,y$ first-order variables and $\varphi(x,y)$ a \fotc formula (which, besides $x$ and
    $y$, possibly contains other free variables). The set of \fotc
    formulas is given by the following recursive definition:
$$\varphi:= At ~|~ Xx ~|~ \varphi \wedge \psi ~|~ \varphi \vee \psi ~|~ \varphi
\rightarrow \psi ~|~ \neg \varphi ~|~ \exists x~ \varphi ~|~ [TC_{xy}
\varphi(x,y)](u,v) ~$$
We use $\forall x \varphi$ as shorthand for $\neg \exists x \neg \varphi$.
We define the \textit{quantifier depth} of a \fotc formula as the maximal
number of nested first-order quantifiers and $TC$ operators.
We interpret \fotc formulas in first-order structures. The notion of assignment and the truth of atomic formulas is defined as in \fo. The truth of compound formulas is defined by induction, with the same clauses as in \fo and an additional one:
\begin{center}
$\frM, g \models [TC_{xy} \varphi](u,v)$\\
iff\\
for all $A
\subseteq M$, if $g(u) \in A$\\
and for all $a, b \in M$, $a\in A$ and $\frM, g[a/x,b/y]
\models \varphi(x,y)$ implies $b\in A$,\\
then $g(v)\in A$.
\end{center}
\label{deftc}
\end{defi}

\begin{prop}
On standard structures, the following semantical clause for the $TC$ operator is equivalent to the one given above:
\begin{center}
$\frM, g \models [TC_{xy}\varphi(x,y)](u,v)$\\
iff\\
there exist $a_1 \ldots a_n \in M$ with $g(u)=a_1$ and $g(v)=a_n$\\
and $\frM, g \models \varphi(a_i,a_{i+1})$ for all $0<i<n$
\end{center}
\label{tcaltsem}
\end{prop}
\begin{proof}
Indeed, suppose there is a finite sequence of points $a_1\ldots
a_n$ such that $g(u)=a_1$, $g(v)=a_n$, and for each $i<n$,
$\frM,g[x/a_i;y/a_{i+1}]\models \varphi(x,y)$. Then for any subset
$A$ closed under $\varphi$ and containing $a_1$, we can show by
induction on the length of the sequence $a_1 \ldots a_n$ that
$a_n$ belongs to $A$. Now, on the other hand, suppose that there
is no finite sequence like described above. To show that there is
a subset $A$ of the required form, we simply take $A$ to be the
set of all points that can be reached from $u$ via $\varphi$ by a
finite sequence. By assumption, $v$ does not belong to this set
and the set is closed under $\varphi$.
\end{proof}

Intuitively this means that for a formula of the form $[\tc_{xy}
\varphi](u,v)$ to hold on a standard structure, there must be a
\emph{finite} ``$\varphi$ path'' between the points that are named
by the variables $u$ and $v$.

Finally we will also be interested in monadic least fixed-point logic
\folfp, which extends \fo with set variables and an explicit
monadic least fixed point operator. Consider a \folfp-formula
$\varphi(X,x)$ and a structure $\frM$ together with a valuation $g$.
This formula induces an operator $F_\varphi$ taking a set $A
\subseteq dom(\frM)$ to the set $\{a:\frM, g[a/x,A/X] \models
\varphi\}$. \folfp is concerned with \emph{least fixed points} of such
operators. If $\varphi$ is positive in $X$ (a formula is positive in
$X$ whenever $X$ only occurs in the scope of an even number of
negations), the operator $F_\varphi$ is monotone (i.e. $X \subseteq
Y$ implies $F_\varphi(X) \subseteq F_\varphi(Y)$). Monotone operators
always have a least fixed point $LFP(F)=\bigcap \{X|F(X)\subseteq X\}$
(defined as the intersection of all
their prefixed points).

\begin{defi}
[Syntax and semantics of \folfp]
    Let $At$ stand for a first-order atomic formula  (of the form
    $R(\vec{x}))$, $x=y$, or $\top$), $X$ a set variable, $x,y$ \fo-variables, $\psi$, $\xi$ \folfp-formulas and $\varphi(x,X)$
    a \folfp-formula positive in $X$ (besides $x$ and $X$, $\varphi(x,X)$ possibly contains other free variables).
    The set of \folfp formulas is given by the following recursive definition:
$$\psi:= At ~|~ Xy ~|~ \psi \wedge \xi ~|~ \psi \vee \xi ~|~ \psi
\rightarrow \xi ~|~ \neg \psi ~|~ \exists x~ \psi ~|~
[LFP_{Xx}\varphi(x,X)]y ~$$
We use $\forall x \psi$ as shorthand for $\neg \exists x \neg
\psi$. We define the \textit{quantifier depth} of a \folfp-formula as the
maximal number of nested first-order quantifiers and $LFP$ operators.
Again, we can interpret \folfp-formulas in first-order structures. The
notion of assignment and the truth of atomic formulas are defined similarly
as in the \mso case. The truth of compound formulas is defined by induction,
with the same clauses as in \fo and an additional one:
\begin{center}
$\frM, g \models [LFP_{Xx}\varphi]y$\\
iff\\
for all $A \subseteq dom(\frM)$, if for all $a \in dom(\frM)$, $\frM, g[a/x,A/X]
\models \varphi(x,X)$ implies $a \in A$,\\
then $g(y)\in A$.
\end{center}
\end{defi}

\begin{rem}
In practice we will use an equivalent (less intuitive but often
more convenient) rephrasing:
\begin{center}
$\frM, g \models [LFP_{Xx}\varphi]y$\\
iff\\
for all $A \subseteq dom(\frM)$, if $g(y) \notin A$,\\
then there exists $a \in dom(\frM)$ such that $a \notin A$ and $\frM, g[a/x,A/X] \models \varphi(x,X)$.
\end{center}
\end{rem}

\subsection{Expressive Power}
\label{exppow}

There is a recursive procedure, transforming any \folfp-formula
$\varphi$ into a \mso-formula $\varphi'$ such that $\frM, g \models
\varphi$ iff $\frM, g \models \varphi'$. The interesting clause is
$([LFP_{Xx}\varphi(x,X)]y)'=\forall X (\forall x (\varphi(x,X)'
\rightarrow Xx) \rightarrow Xy)$. (The other ones are all of the same type, e.g. $(\varphi \wedge \psi)'=(\varphi' \wedge \psi')$.) This procedure can easily be seen adequate by considering the semantical clause for the $LFP$ operator.

Now there is also a recursive procedure transforming any \fotc formula
$\varphi$ into a \folfp formula $\varphi''$ such that $\frM, g \models
\varphi$ iff $\frM, g \models \varphi''$. The interesting clause is
$([TC_{xy} \varphi](u,v))''=[LFP_{Xy} y=u \vee \exists x((Xx \wedge \varphi(x,y)''))]v$.
Let us give an argument for this claim. By Proposition \ref{tcaltsem} it is
enough to show that $[LFP_{Xy} y=u \vee \exists x(Xx \wedge
\varphi(x,y)'')]v$ holds if and only if there is a finite
$\varphi''$ path from $u$ to $v$. For the right to left direction,
suppose there is such a path $a_1 \ldots a_n$ with $g(u)=a_1$ and
$g(v)=a_n$. Then, for any subset $A$ of the domain, we can show by
induction on $i$ that if for all $a_i$ ($1 \leq i \leq n)$, $a_i=u
\vee \exists x((Ax \wedge \varphi(x,a_i)'')$ implies $a_i \in A$,
then $v \in A$, i.e., $[LFP_{Xy} y=u \vee \exists x((Xx \wedge
\varphi(x,y)''))]v$ holds. Now for the left to right direction,
suppose there is no such $\varphi''$ path. Consider the set $A$ of
all points that can be reached from $u$ by a finite $\varphi''$
path. By assumption, $\neg Av$ and it holds that $\forall y ((y=u
\vee \exists x(Ax \wedge \varphi(x,y)''))\rightarrow Ay)$, i.e.,
$\neg[LFP_{Xy} y=u \vee \exists x(Xx \wedge \varphi(x,y)'')]v$.

It is known that on arbitrary structures $\fotc<\folfp<\mso$ (see \cite{ebfl95} or \cite{leo}) and on trees $\fotc<_{trees}\folfp=_{trees}\mso$ (see \cite{1376952} and \cite{DBLP:conf/icalp/Schweikardt04}). It is also known that the (not \fo definable) class of finite trees is already definable in \fotc (see for instance \cite{2006}), which is the weakest of the logics studied here. We provide additional detail in Section \ref{sec53}.

\section{The Axiomatizations}
\label{sec2}

\begin{figure}[!h]
\hrule\smallskip
\begin{tabular}{@{}ll@{\hspace{14mm}}l@{}}
FO1. & $\vdash\phi$, whenever $\phi$ is a propositional tautology \\
FO2. & $\vdash \forall x \varphi \rightarrow \varphi^x_t$, where $t$ is substitutable for $x$ in $\varphi$\\
FO3. & $\vdash \forall x (\varphi \rightarrow \psi) \rightarrow (\forall x \varphi \rightarrow \forall x \psi)$\\
FO4. & $\vdash \varphi \rightarrow \forall x \varphi$, where $x$ does not occur free in $\varphi$\\
FO5. & $\vdash x=x$\\
FO6. & $\vdash x=y \rightarrow (\varphi \rightarrow \psi)$, where $\varphi$ is atomic and $\psi$ is obtained\\
 & from $\varphi$ by replacing $x$ in zero or more (but not necessarily\\
 & all) places by $y$.\\
Modus Ponens & if $\vdash\varphi$ and $\vdash \varphi \rightarrow \psi$, then $\vdash\psi$\\
FO Generalization & if $\vdash \varphi$, then $\vdash \forall x \varphi$\\
\end{tabular}
\smallskip\hrule
\caption{Axioms and rules of \fo} \label{fig:FO-axioms}
\end{figure}

\begin{figure}[!h]
\hrule\smallskip
\begin{tabular}{@{}ll@{\hspace{14mm}}l@{}}
COMP. & $\vdash \exists X \forall x (Xx \leftrightarrow \varphi)$, where $X$ does not occur free in $\varphi$\\
MSO1. & $\vdash \forall X \varphi \rightarrow \varphi[X/T]$, where $T$ (which is either a set variable\\
 & or a set predicate) is substitutable in $\varphi$ for $X$.\\
MSO2. & $\vdash \forall X (\varphi \rightarrow \psi) \rightarrow (\forall X \varphi \rightarrow \forall X \psi)$\\
MSO3. & $\vdash \varphi \rightarrow \forall X \varphi$, where $X$ does not occur free in $\varphi$\\
MSO Generalization & if $\vdash \varphi$, then $\vdash \forall X \varphi$\\
\end{tabular}
\smallskip\hrule
\caption{Axioms and inference rule of \mso} \label{fig:MSO-axioms}
\end{figure}

\begin{figure}[!h]
\hrule\smallskip
\begin{tabular}{@{}ll@{\hspace{14mm}}l@{}}
\fotc & $\vdash [TC_{xy} \varphi](u,v) \rightarrow ((\psi(u) \land \forall x \forall y (\psi(x)\land\varphi(x,y)\to \psi(y))) \to \psi(v))$\\
axiom & where $\psi$ is any \fotc formula\\[2mm]
\fotc & if $\vdash \xi \rightarrow ((P(u) \land \forall x \forall y (P(x)\land\varphi(x,y)\to P(y))) \to P(v))$,\\
 Genera- & and $P$ does not occur in $\xi$,\\
lization & then $\vdash \xi \rightarrow [TC_{xy}\varphi](u,v)$\\
\end{tabular}
\smallskip\hrule
\caption{Axiom and inference rule of \fotc} \label{fig:fotc-axioms}
\end{figure}

\begin{figure}[!h]
\hrule\smallskip
\begin{tabular}{@{}ll@{\hspace{14mm}}l@{}}
\folfp & $\vdash [LFP_{Xx}\varphi]y \rightarrow (\forall x (\varphi(x,\psi)\rightarrow \psi(x)) \rightarrow \psi(y))$\\
axiom & where $\psi$ is any \folfp formula and $\varphi(x,\psi)$ is the result\\
 & of the replacement in $\varphi(x,X)$ of each occurrence of $X$ by $\psi$\\
 & (renaming variables when needed)\\[2mm]
\folfp & if $\vdash \xi \rightarrow (\forall x (\varphi(x,P) \rightarrow P(x)) \rightarrow P(y))$,\\
 Generalization & and $P$ positive in $\varphi$ does not occur in $\xi$,\\
  & then $\vdash \xi \rightarrow [LFP_{Xx} \varphi](y)$\\
\end{tabular}
\smallskip\hrule
\caption{Axiom and inference rule of \folfp} \label{fig:folfp-axioms}
\end{figure}

\begin{figure}[!h]
\hrule\smallskip
\begin{tabular}{@{}lll@{}}
T1. & $\forall x \forall y \forall z(x<y \land y<z \to x<z)$ & $<$ is transitive \\[2mm]
T2. & $\neg \exists x (x< x)$ & $<$ is irreflexive \\[2mm]
T3. &$\forall x \forall y(x<y \to \exists z (x<_{ch}z \land z\leq y))$ & immediate child\\[2mm]
T4. &$\exists x \forall y (x \leq y)$ & there is a unique root \\[2mm]
T5. &$\forall x \forall y \forall z (x < z \land y < z \to x \leq y \lor y \leq x)$ &
  linearly ordered  branches\\[2mm]
T6. &$\forall x \forall y \forall z(x\prec y \land y\prec z \to x\prec z)$ & $\prec$ is transitive\\[2mm]
T7. &$\neg \exists x (x\prec  x)$ & $\prec $ is irreflexive \\[2mm]
T8. &$\forall x \forall y(x\prec y \to \exists z(x\prec_{ns}z \land z\preceq y))$ & immediate next sibling\\[2mm]
T9. &$\forall x\exists y(y\preceq x \land \neg \exists z(z\prec y))$ &
  there is a least sibling\\[2mm]
T10. & $\forall x \forall y((x\prec y \lor y \prec x)\leftrightarrow
            (\exists z(z<_{ch}x \land z<_{ch}y) \land x\neq y))$ &
   linearly ordered siblings\\[2mm]
Ind. &  $\forall x(\forall y((x <y \lor x\prec y) \to \varphi(y)) \to
\varphi(x)) \to \forall x\varphi(x)$& induction scheme\\
\\
where\!\!\!\!\!\!\!\\
& $\varphi(x)$ ranges over $\Lambda$-formulas in one free variable $x$, \\
and \\
 & \multicolumn{2}{l}{$x<_{ch}y$ is shorthand for $x<y \land \neg
  \exists z(z<y\wedge x<z)$,} \\
  & \multicolumn{2}{l}{$x\prec_{ns}y$ is shorthand for $x \prec y \land \neg \exists z(x\prec z\wedge z\prec y)$} \\
\end{tabular}
\smallskip\hrule
\caption{Specific axioms on finite trees} \label{fig:tree-axioms}
\end{figure}

As many arguments in this paper equally hold for \mso, \fotc and
\folfp, we let $\Lambda \in \{\mso, \fotc, \folfp\}$ and use
$\Lambda$ as a symbol for any one of them. The axiomatization of
$\Lambda$ on finite trees consists of three parts: the axioms of
first-order logic, the specific axioms of $\Lambda$, and the
specific axioms on finite trees.

To axiomatize \fo, we adopt the infinite set of
logical axioms and the two rules of inference given in Figure
\ref{fig:FO-axioms} (like in \cite{Enderton}, except from the fact
that we use a generalization rule). Here, as in \cite{Enderton}, by a 
propositional tautology, we mean a formula can be obtained
from a valid propositional formula (also known as the 
sentential calculus) by uniformly substituting formulas for the 
proposition letters). Alternatively, FO1 may be replaced by a 
complete set of axioms for propositional logic.
To axiomatize \mso, the
axioms and rule of Figure \ref{fig:MSO-axioms} are added to the
axiomatization of \fo and we call the resulting system $\vdash_\mso$.
COMP stands for ``comprehension'' by analogy with the
comprehension axiom of set theory. MSO1 plays a similar role as
FO2, MSO2 as FO3 and MSO3 as FO4. To axiomatize \fotc,
the axiom and rule of Figure \ref{fig:fotc-axioms} are added to
the axiomatization of \fo and we call the resulting system
$\vdash_\fotc$. To axiomatize \folfp, the axiom and rule of Figure
\ref{fig:folfp-axioms} are added to the axiomatization of \fo and we
call the resulting system $\vdash_\folfp$. We are interested in
axiomatizing $\Lambda$ on the class of finite trees. For that
purpose, we restrict the class of considered structures by adding
to $\vdash_\Lambda$ the axioms given in Figure
\ref{fig:tree-axioms} and we call the resulting system
$\vdash_\Lambda^{tree}$. Note that the induction scheme in Figure
\ref{fig:tree-axioms} allows to reason by induction on \emph{properties definable in $\Lambda$} only.

\begin{prop}\label{prop:define-finite}
A finite structure $\mathfrak{M}=(M,<,\prec,P_1, \ldots, P_n)$ satisfies the axioms T1--T10 if and only
if $\mathfrak{M}$ is a finite tree.
\end{prop}

\begin{proof}
  It follows from the truth of T1, T2 and T5 that $(M,<)$ is a tree.
  Note that T3 and T4 are valid consequences of T1, T2, T5 on finite
  structures.  Furthermore, T6, T7 and T10 imply that $\prec$
  linearly orders the children of each node (and that $\prec$ only
  relates to each other nodes that are siblings). Note again that
  T8 and T9 follow from T6, T7 and T10 on finite structures.
\end{proof}

In fact, as we will see later, cf. Theorem~\ref{thm:define}, the
axioms T1--T10, together with the induction scheme Ind for
$\Lambda$-formulas (where $\Lambda\in\{\mso, \fotc, \folfp\}$) define
the class of finite trees.

We refer for basic definitions (e.g., \emph{proof} by which we mean \emph{formal deduction}, or \emph{axiomatization} by which we mean \emph{deductive calculus}) to \cite{Enderton} and sometimes only sketch or even omit classical arguments. E.g., we assume the notion of \emph{being substitutable} in a formula
to be clear for both objects and set variables. For details on such
basic notions and technics, we refer to the material extensively
developed in \cite{Enderton} and in particular, to the proof of the \fo
completeness theorem presented there. The Henkin completeness proofs
provided in Section \ref{sec3} are built on this classical material. 

We end this Section by spelling out some definitions that are specific to our paper.

\begin{defi}
We say that a $\Lambda$-formula $\varphi$ is \emph{$\Lambda$-provable} if $\vdash_\Lambda \varphi$
occurs (as the last line) in some $\Lambda$-proof and we say that it is \emph{$\Lambda$-consistent} if its negation is not $\Lambda$-provable. 

Let $\Gamma$ be a set of $\Lambda$-formulas and $\varphi$ a
$\Lambda$-formula. By $\Gamma \vdash_\Lambda \varphi$ we will
always mean that there are $\psi_1,\ldots, \psi_n \in \Gamma$ such
that $\vdash_\Lambda (\psi_1 \wedge \ldots \wedge \psi_n)
\rightarrow \varphi$. Whenever $\Gamma \vdash_\Lambda \neg \varphi$
does not hold, we say that $\varphi$ is $\Gamma$-consistent. We say that $\Gamma$ is $\Lambda$-consistent if $\top$ is $\Gamma$-consistent.
Finally, we say that $\Gamma$ is a \emph{maximal consistent set} of $\Lambda$-formulas if 
$\Gamma$ is consistent, and for each formula $\phi\in\Lambda$, either
$\phi\in\Gamma$ or $\neg\phi\in\Gamma$.

\label{defax}
\end{defi}

Now the main result of this paper is that on standard structures,
the $\Lambda$ theory of finite trees is completely axiomatized by
$\vdash_\Lambda^{tree}$. In the remaining sections we will
progressively build a proof of it.

\section{Henkin Completeness}
\label{sec3}

As it is well known, \mso, \fotc and \folfp are highly undecidable on arbitrary standard
structures and hence not recursively enumerable (by arbitrary, we mean 
when there is no restriction on the interpretation of the relation
symbols  from the signature, unlike in the case of, e.g., trees). 
So in order to show that our axiomatizations $\vdash_\Lambda^{tree}$ are complete on finite trees, we refine a trick used by
Kees Doets in his PhD thesis \cite{1987}. We proceed in two steps (the second step being the one inspired by Kees Doets). First, we
show completeness theorems, based on a non-standard (so
called Henkin) semantics for \mso, \fotc and \folfp (on the general topic of Henkin semantics, see \cite{1950}, the original paper by Henkin and also \cite{230876}). Each
semantics respectively extends the class of standard structures with
non standard (Henkin) \mso, \fotc and \folfp-structures. By the
Henkin completeness theorems, our axiomatic systems
$\vdash_\Lambda^{tree}$ naturally turn out to be complete on the
wider class of their Henkin-models. But we will see that compactness also follows from these completeness results and some of these Henkin
models are infinite. As a second step, we show in Section \ref{sec5} that \emph{no
$\Lambda$-sentence can distinguish between standard and non-standard $\Lambda$-Henkin-models among models of our axioms}. Every finite Henkin model being also a standard model, this entails that our axioms are complete on the class of (standard) finite trees, i.e., each $\Lambda$-sentence valid on this class is provable using the system
$\vdash_\Lambda^{tree}$. 

Now let us point out that Kees Doets was interested in complete axiomatizations of monadic ``$\Pi_1^1$-theories'' of various classes of
linear orders and trees. Considering such theories in fact amounts to considering \emph{first-order} theories of such structures extended with finitely many unary predicates. Thus, he was relying on the \fo completeness theorem and if he was working with non-standard models of particular \fo-theories, he was not concerned with non standard Henkin-structures in our sense. In particular, he used Ehrenfeucht-Fra\"{i}ss\'{e} games in order to show that ``definably well-founded'' node-labeled trees have well-founded $n$-equivalents for all $n$. 
In Section \ref{sectproof}, Lemma \ref{lemma3}, which is the key lemma to our main completeness result, establishes a similar result for definably well-founded Henkin-models of the $\Lambda$-theory of finite node-labeled sibling-ordered finite trees. Hence, what makes the originality of the method developed in this paper is its use of Henkin semantics: we first create a Henkin model and then ``massage'' it in order to obtain a model that is among our intended ones. Similar methods are commonly used to show completeness results in modal logic, where ``canonical models'' are often transformed in order to obtain intended models (see \cite{modal}). Remarkably, the completeness proof for the $\mu$-calculus on finite trees given in \cite{DBLP:conf/fossacs/CateF10}, which is directly inspired by the methods used here, proceeds in that way. There are numerous examples of that sort in modal logic (and especially, in temporal logic), but there is also one notable example in classical model theory. In 
1970, Keisler provided a complete axiomatization of \fo extended with the quantifier ``there exist uncountably many'' (see \cite{keisler}). His completeness proof, which is established for standard models, is surprisingly simple, it relies on the construction of an elementary chain of Henkin structures and then uses the omitting types theorem. Hence all in all, these structures seem to provide a particularly convenient tool, not only for simple Henkin completeness proofs, but also for more refined completeness proofs with respect to interesting subclasses of Henkin models like standard models. 

Let us now introduce Henkin structures formally. Such structures are
particular cases among structures called \emph{frames} (note that such frames are unrelated to ``Kripke frames'') and it is convenient to define frames before defining Henkin-structures. In our case, a frame is simply a relational structure together with some subset of the powerset of its domain called its \emph{set of admissible subsets}. A Henkin structure is a frame whose set of admissible subsets satisfies some natural closure conditions.

\begin{defi}
[Frames]
    Let $\sigma$ be a purely relational vocabulary. A \textbf{$\sigma$-frame} $\frM$ consists of a non-empty domain
    $dom(\frM)$, an interpretation in $dom(\frM)$ of the predicates in $\sigma$ and a set of admissible subsets $\mathbb{A}_\frM \subseteq \wp(dom(\frM))$.
\end{defi}
Whenever $\mathbb{A}_\frM = \wp(dom(\frM))$, $\frM$ can be identified to a standard structure. Assignments $g$ into $\frM$ are defined as in standard semantics, except that if $X$ is a set variable, then we require that $g(X) \in
    \mathbb{A}_\frM$.
\begin{defi}
[Interpretation of $\Lambda$-formulas in frames]
$\Lambda$-formulas are interpreted in frames as in standard structures, except for the three following clauses. The set quantifier clause of \mso becomes:
\begin{center}
$\frM, g \models \exists X \varphi$ iff there is $A \in
\mathbb{A}_{\frM}$ such that $\frM, g[A/X] \models \varphi$
\end{center}
The $TC$ clause of \fotc becomes:
\begin{center}
$\frM, g \models [TC_{xy} \varphi](u,v)$\\
iff\\
for all $A \in \mathbb{A}_{\frM}$, if $g(u) \in A$\\
and for all $a, b \in dom(\frM)$, $a\in A$ and $\frM, g[x/a,b/y]
\models \varphi$ imply
$b\in A$,\\
then $g(v)\in A$.
\end{center}

And finally the $LFP$ clause of \folfp becomes:
\begin{center}
$\frM, g \models [LFP_{Xx}\varphi]y$\\
iff\\
for all $A \in \mathbb{A}_\frM$, if for all $a \in dom(\frM)$, $\frM, g[a/x,A/X]
\models \varphi(x,X)$ implies $a \in A$,\\
then $g(y)\in A$.
\end{center}
\label{hensem}
\end{defi}

\begin{defi}
[$\Lambda$-Henkin-Structures]
    A $\Lambda$-\textbf{Henkin-structure} is a frame $\frM$ that is closed
    under parametric $\Lambda$-definability, i.e., for each $\Lambda$-formula
    $\varphi$ and assignment $g$ into $\frM$:
\begin{center}
$\{a \in M ~|~ \frM, g[a/x] \models \varphi\} \in \mathbb{A}_{\frM}$
\end{center}
We call a $\Lambda$-Henkin-structure $\frM$ \emph{standard} whenever
every subset in $dom(\frM)$ belongs to $\mathbb{A}_{\frM}$. 
\end{defi}

\begin{rem}
Note that any finite $\Lambda$-Henkin-structure is a standard
structure, as every subset of the domain is parametrically
definable in a finite structure. Hence, non standard Henkin structures are
always infinite.
\label{rem}
\finremark
\end{rem}

\begin{thm}
$\Lambda$ is completely axiomatized on $\Lambda$-Henkin-structures
by $\vdash_\Lambda$, in fact for every set of $\Lambda$-formulas
$\Gamma$ and $\Lambda$-formula $\varphi$, $\varphi$ is true in all
$\Lambda$-Henkin-models of $\Gamma$ if and only if $\Gamma
\vdash_\Lambda \varphi$. \label{hencomp}
\end{thm}

We do not detail here the \mso proof, as it is a special case of the proof of completeness for the theory of types given in \cite{230876}. We focus only on the \folfp case, as the \fotc case is very similar, except that there is no need to consider set variables. Up to now we have been working with purely relational vocabularies. Here we will be using individual constants in the standard way, but only for the sake of readability (we could dispense with them and use \fo variables instead). Also, whenever this is clear from the context, we will use $\vdash$ as shorthand for $\vdash_\folfp$. Let us now begin the Henkin completeness proof for \folfp. This will achieve the proof of Theorem \ref{hencomp}.

\begin{lem}[Generalization Lemma for \fo Quantifiers]
If $\Gamma \vdash_\folfp \varphi$ and $x$ does not occur free in
$\Gamma$, then $\Gamma \vdash_\folfp \forall x \varphi$. \label{fogen}
\end{lem}

\begin{proof}
We refer the reader to the proof for \fo given by Enderton in
\cite[page 117]{Enderton}. The same proof applies for \folfp (as well
as \mso and \fotc).
\end{proof}

\begin{defi}
\label{folfpwit} We say that a set of \folfp formulas $\Delta$
contains \folfp Henkin witnesses if and only if the two following
conditions hold. First, for every formula $\varphi$, if
$\neg\forall x \varphi \in \Delta$, then $\neg \varphi[x/t] \in
\Delta$ for some term $t$ and if $\neg[LFP_{Xx}\varphi]y \in
\Delta$, then $\neg Py \wedge \neg \exists x (\neg Px \wedge
\varphi(P,x)) \in \Delta$ for some monadic predicate $P$.
Second, if $\varphi \in \Delta$ and $x$ is a free variable of
$\varphi$, then $\forall x(Px\leftrightarrow \varphi(x)) \in
\Delta$ for some monadic predicate $P$.
\end{defi}

The originality of the \folfp case essentially lies in the notion of \folfp-Henkin witness of Definition \ref{folfpwit}. 
In order to use this notion in the proof of Lemma \ref{folfplind}, we also need the following lemma:

\begin{lem}
Let $\Gamma$ be a consistent set of \folfp-formulas and $\theta$ a
\folfp-formula of the form $\forall x (\varphi \leftrightarrow
Px)$ with $P$ a fresh monadic predicate (i.e. not appearing in
$\Gamma$). Then $\Gamma \cup \{\theta\}$ is also consistent.
\label{lemma6}
\end{lem}
\begin{proof}
Suppose $\Gamma \cup \{\forall x (\varphi \leftrightarrow Px)\}$
is inconsistent, so there is some proof of $\bot$ from formulas in
$\Gamma \cup \{\forall x(\varphi \leftrightarrow Px)\}$. We first
rename all bound variables in the proof with variables which had
no occurrence in the proof or in $\forall x(\varphi\leftrightarrow
Px)$ (this is possible since proofs are finite objects and we have
a countable stock of variables). Also, whenever in the proof the
\folfp generalization rule is applied on some unary predicate $P$,
we make sure that this $P$ is different from the unary predicate
that we want to substitute by $\varphi$ and which does not appear
in the proof; this is always possible because we have a countable
set of unary predicates. Now, we replace in the proof all
occurrences of $Px$ by $\varphi$ (as we renamed bound variables,
there is no accidental binding of variables by wrong quantifiers).
Then, every occurrence of $\forall x(\varphi \leftrightarrow Px)$
in the proof becomes an occurrence of $\forall x (\varphi
\leftrightarrow \varphi)$, i.e., we have obtained a proof of
$\bot$ from $\Gamma \cup \{\forall x(\varphi\leftrightarrow
\varphi)\}$, i.e., from $\Gamma$ ($\forall x (\varphi
\leftrightarrow \varphi)$ is provable, as it can be obtained by
\fo generalization from a propositional tautology). It
entails that $\Gamma$ is already inconsistent, which contradicts
the consistency of $\Gamma$. Now it remains to show that the
replacement procedure of all occurrences of $Px$ by $\varphi$ is
correct, so that we still have a proof of $\bot$ after it. Every
time the replacement occurs in an axiom (or its generalization,
which is still an axiom as we defined it), then the result is
still an instance of the given axiom schema (even for \folfp
generalizations, because we took care that $P$ is never used in
the proof for a \folfp generalization). Also, as replacement is
applied uniformly in the proof, every application of modus ponens
stays correct: consider $\psi \rightarrow \xi$ and $\psi$.
Obviously the result $\psi^*$ of the substitution will allow to
derive the result $\xi^*$ of the substitution from $\psi^*
\rightarrow \xi^*$ and $\psi^*$. Also $\bot^*$ is simply $\bot$,
so the procedure gives us a proof of $\bot$.
\end{proof}

\begin{lem}(\folfp Lindenbaum Lemma)
Let $\sigma$ be a countable vocabulary and let $\sigma^*=\sigma \cup \{c_n ~|~ \in \mathbb{N}\} \cup \{P_n
~|~ n \in \mathbb{N}\}$ with $c_i, P_i \notin \sigma$. If a set $\Gamma$ of \folfp-formulas in vocabulary $\sigma$ is consistent, then there exists a maximal
consistent set $\Gamma^*$ of $\sigma^*$ formulas such that $\Gamma
\subseteq \Gamma^*$ and $\Gamma^*$ contains \folfp-Henkin
witnesses.
\label{folfplind}
\end{lem}

\begin{proof}
Let $\Gamma$ be a consistent set of well formed \folfp-formulas in
a countable vocabulary $\sigma$.
We expand $\sigma$ to $\sigma^*$ by adding countably many new constants and countably many new monadic predicates.
Then $\Gamma$ remains consistent as a set of well formed formulas
in the new language. We fix an enumeration of all tuples consisting of two \fo variables, one set variable and one formula
of $\sigma^*$:
\begin{center}
$<\varphi_1, x_1, x_1', X_1>,<\varphi_2, x_2, x_2', X_2>,<\varphi_3, x_3, x_3', X_3>,
\ldots$
\end{center}
(this is possible since the language is countable), where the $\varphi_i$
are formulas, the $x_i$, $x_i'$ are \fo variables and the $X_i$ are set variables.
\begin{iteMize}{$\bullet$}
\item   Let $\theta_{3n-2}$ be $\neg \forall x_n \varphi_n \rightarrow \neg \varphi_n[x_n/c_l]$,
where $c_l$ is the first of the new constants neither occurring in
$\varphi_n$ nor in $\theta_k$ with $k < 3n-2$.
\item   Let $\theta_{3n-1}$ be $\neg [LFP_{x_n'X_n}\varphi_n]x_n \rightarrow (\neg P_lx_n \wedge \neg
\exists x (\neg P_lx \wedge \varphi_n(P_l,x)))$, where $P_l$ is the
first of the new monadic predicates neither occurring in
$\varphi_n$ nor in $\theta_k$ with $k < 3n-1$.
\item   Let $\theta_{3n}$ be $\forall x_n (\varphi_n \leftrightarrow P_lx_n)$, where $P_l$ is the first of the new monadic predicates neither occurring in $\varphi_n$ nor in
$\theta_k$ with $k < 3n$.
\end{iteMize}

\noindent Call $\Theta$ the set of all the $\theta_i$.

\begin{clm}
$\Gamma \cup \Theta$ is consistent
\end{clm}
If not, then because deductions are finite, for some $m \geq 0$,
$\Gamma \cup \{\theta_1, \ldots, \theta_m,\theta_{m+1}\}$ is
inconsistent. Take the least such $m$. Then, by the
definition of consistency and the axioms of propositional logic,
$\Gamma \cup \{\theta_1, \ldots, \theta_m\} \vdash
\neg\theta_{m+1}$. Now there are three cases:
\begin{enumerate}[(1)]
\item    $\theta_{m+1}$ is of the form $\neg \forall x \varphi \rightarrow
\neg \varphi[x/c]$, so both $\Gamma \cup
\{\theta_1,\ldots,\theta_m\} \vdash \neg \forall x \varphi$ and
$\Gamma \cup \{\theta_1,\ldots,\theta_m\} \vdash \varphi[x/c]$.
Since $c$ does not appear in any formula on the left, by Lemma
\ref{fogen}, $\Gamma \cup \{\theta_1,\ldots,\theta_m\} \vdash
\forall x \varphi$, which contradicts the minimality of $m$ (or
the consistency of $\Gamma$ if $m=0$).
\item   $\theta_{m+1}$ is of the form $\neg [LFP_{Xx}\varphi]y \rightarrow (\neg Py
\wedge \neg \exists x (\neg Px \wedge \varphi(P,x)))$. In such a
case both $\Gamma \cup \{\theta_1 \ldots \theta_m\} \vdash \neg
[LFP_{Xx}\varphi]y$ and $\Gamma \cup \{\theta_1 \ldots \theta_m\}
\vdash \neg(\neg Py \wedge \neg \exists x (\neg Px \wedge
\varphi(P,x)))$ hold. It follows that  $\Gamma \cup \{\theta_1 \ldots \theta_m\}
\vdash \forall x (\varphi(P,x) \rightarrow P(x)) \rightarrow Py$. Since $P$ does not appear in any formula on
the left, by \folfp generalization, $\Gamma \cup \{\theta_1 \ldots
\theta_m\} \vdash [LFP_{Xx}\varphi]y$, which contradicts the
minimality of $m$ (or the consistency of $\Gamma$ whenever $m=0$).
\item   $\theta_{m+1}$ is of the form $\forall x (\varphi \leftrightarrow Px)$.
By Lemma \ref{lemma6}, this is not possible.
\end{enumerate}
We extend $\Gamma \cup \Theta$ to a maximal consistent set
$\Gamma^*$ in the standard way (see for instance \cite[page 137]{Enderton}).
\end{proof}

We will now show that if $\Gamma^*$ is a maximal consistent set that
contains \folfp-Henkin witnesses, then $\Gamma^*$ has a \folfp-Henkin
model $\frM_{\Gamma^*}$.

\begin{defi}
Let $\Gamma^* \subseteq FORM(\sigma)$ be maximal consistent and
contain \folfp-Henkin witnesses. We define an equivalence relation on the
set of \fo terms, by letting $t_1\equiv_{\Gamma^*}t_2$ iff
$t_1=t_2 \in \Gamma^*$. We denote the equivalence class of a term
$t$ by $|t|$.
\end{defi}
\begin{prop}
 $\equiv_{\Gamma^*}$ is an equivalence relation.
\end{prop}
\begin{proof}\
By  FO5 and FO6.
\end{proof}

\begin{defi}
We define $\frM_{\Gamma^*}$ (together with a valuation $g_{\Gamma^*}$) out of $\Gamma^*$.
\begin{iteMize}{$\bullet$}
\item   $M=\{|t|:t$ is a \fo term $\}$
\item   $\mathbb{A}_{\frM_{\Gamma^*}}=\{A_T:T$ is a set variable or a monadic predicate$\}$ where $A_T=\{|t|:Tt\in\Gamma^*\}$
\item   $(|t_1|,\ldots,|t_n|) \in P^\frM_{\Gamma^*}$ iff $Pt_1 \ldots t_n \in \Gamma^*$
\item   $c^{\frM_{\Gamma^*}}=|c|$
\item   $g_{\Gamma^*}(x)=|x|$
\item   $g_{\Gamma^*}(X)=A_X$
\end{iteMize}
\end{defi}

\begin{prop}
 $\frM_{\Gamma^*}$ is a \folfp-Henkin structure.
\end{prop}

\begin{proof}
By construction of $\Gamma^*$ which contains \folfp-Henkin witnesses, this is immediate (we introduced a
monadic predicate for each parametrically definable subset).
\end{proof}

\begin{lem}(Truth lemma) For every \folfp formula $\varphi$, $\frM_{\Gamma^*}, g_{\Gamma^*} \models \varphi$ iff $\varphi \in \Gamma^*$.
\end{lem}
\begin{proof}
By induction on $\varphi$.

The base case follows from the definition of $\frM_{\Gamma^*}$ together with the maximality of $\Gamma^*$.
Now consider the inductive step:
\begin{iteMize}{$\bullet$}
\item   Boolean connectives and FO quantifier: exactly as in \fo (see
  \cite[page 138]{Enderton}), basically, for the \fo quantifier step we rely on the fact that $\Gamma^*$ contains \folfp~ Henkin witnesses and we use the $\theta_{3n+2}$ formulas introduced in the proof of Lemma \ref{folfplind}.
\item   $LFP$ operator: we want to show that
\begin{center}

$\frM_{\Gamma^*}, g_{\Gamma^*} \models [LFP_{Xx}\varphi]y$ iff
$[LFP_{Xx}\varphi]y\in\Gamma^*$
\end{center}
\begin{iteMize}{$-$}
\item   We first show that
\begin{center}
$\frM_{\Gamma^*}, g_{\Gamma^*} \models
    [LFP_{Xx}\varphi]y$ implies
    $[LFP_{Xx}\varphi]y\in\Gamma^*$.
    \end{center}
    Assume $\frM_{\Gamma^*}, g_{\Gamma^*}
    \models [LFP_{Xx}\varphi]y$, i.e., for all monadic predicate in $\sigma^*$ or set variable $T$,
    if $g_{\Gamma^*}(y) \notin A_{T}$ then there exists $|t| \in M$, such that $|t| \notin A_{T}$ and
    $\frM_{\Gamma^*}, g_{\Gamma^*}[x/|t|,X/A_{T}] \models \varphi$. It follows by induction hypothesis that
    for every such $T$, if $Ty \not \in \Gamma^*$, then there exists a term $t$ such that $Tt \not \in \Gamma^*$ and
    $\varphi(t,T) \in \Gamma^*$. By maximal consistency of $\Gamma^*$ and using the contraposition of the FO2 axiom, it follows that  
    for all monadic predicate in $\sigma^*$ or set variable $T$ such that $Ty \not \in \Gamma^*$, it holds that
    $\exists x (\neg Tx \wedge \varphi(x,T)) \in \Gamma^*$.
    Now suppose $[LFP_{Xx}\varphi]y \notin \Gamma^*$. By maximal consistency of $\Gamma^*$, we get $\neg[LFP_{Xx}\varphi]y \in \Gamma^*$.
    Then as $\Gamma^*$ contains \folfp~ Henkin witnesses, there is a predicate $T$ such that for some $n$, $\theta_{3n-1} \in \Gamma^*$ is of the form $\neg[LFP_{Xx}\varphi]y \rightarrow (\neg Ty \wedge \neg \exists x(\neg Tx \wedge \varphi(T,x)))$. By maximal consistency of $\Gamma^*$, it follows that $\neg Ty \wedge \neg \exists x(\neg Tx \wedge \varphi(T,x)) \in \Gamma^*$. Hence there is a predicate $T$ such that $Ty \not\in \Gamma^*$ and $\neg \exists x (\neg Tx \wedge \varphi(x,T)) \in \Gamma^*$. But that contradicts the consistency of
    $\Gamma^*$, as we previously showed that whenever $Ty \not\in \Gamma^*$, then also  $\exists x (\neg Tx \wedge \varphi(x,T)) \in \Gamma^*$. Then $\neg[LFP_{Xx}\varphi]y \not\in \Gamma^*$ and by maximal consistency of $\Gamma^*$, $[LFP_{Xx}\varphi]y \in \Gamma^*$.
\item   We now show that $[LFP_{Xx}\varphi]y \in\Gamma^*$ implies
    $\frM_{\Gamma^*}, g_{\Gamma^*} \models
    [LFP_{Xx}\varphi]y$. We consider the contraposition
      \begin{center}
      $\frM_{\Gamma^*}, g_{\Gamma^*} \not\models [LFP_{Xx}\varphi]y$ implies
    $[LFP_{Xx}\varphi]y\not\in\Gamma^*$.
    \end{center}
    Assume $\frM_{\Gamma^*}, g_{\Gamma^*} \not\models [LFP_{Xx}\varphi]y$. So $\frM_{\Gamma^*},
    g_{\Gamma^*} \models \neg [LFP_{Xx}\varphi]y$ and there exists a monadic predicate in $\sigma^*$ or a set variable $T$ such that $A_T \in \mathbb{A}_{\frM_{\Gamma^*}}$, $g(y) \notin A_T$ and
    for all $|t| \in M$, $|t|\in A_T$ or $\frM_{\Gamma^*}, g_{\Gamma^*}[x/|t|,X/T]
    \models \neg \varphi$. By induction hypothesis $Ty \not\in \Gamma^*$ and for all term $t$, either $Tt \in \Gamma^*$, or $\neg \varphi(t,T) \in \Gamma^*$. By maximal consistency of $\Gamma^*$, for all term $t$,
    $Tt \vee \neg \varphi(t,T) \in \Gamma^*$. Now assume $\neg \forall x(Tx \vee \neg \varphi(x,T)) \in \Gamma^*$. As $\Gamma^*$ contains Henkin witnesses, there is some $n$ and some term $t$ such that $\theta_{3n-2} \in \Gamma^*$ is of the  form $\neg \forall x (Tx \vee \neg \varphi(x,T))\rightarrow (Tt \vee \neg \varphi(t,T))$ and hence $Tt \vee \neg \varphi(t,T) \in \Gamma^*$, which contradicts the maximal consistency of $\Gamma^*$. Hence $\forall x(Tx \vee \neg \varphi(x,T)) \in \Gamma^*$.
 By maximal consistency of $\Gamma^*$, $\neg Ty \wedge \forall x (Tx
    \vee \neg \varphi(T,x)) \in \Gamma^*$ and so also $\neg Ty \wedge \neg \exists x(\neg Tx \wedge \varphi(T,x)) \in \Gamma^*$. Now suppose $[LFP_{xX}\varphi]y \in \Gamma^*$. Then by the LFP axiom,
    for every monadic predicate  in $\sigma^*$ or set variable $T$, we get that $\neg Ty \rightarrow \exists x(\neg T(x) \wedge \varphi(x,T)) \in \Gamma^*$ and so $\neg(\neg Ty \wedge \neg \exists x(\neg Tx \wedge \varphi(T,x))) \in \Gamma^*$. But that
    contradicts the maximal consistency of $\Gamma^*$. 
   \end{iteMize}
\end{iteMize}
\end{proof}

\begin{thm}
Every consistent set $\Gamma$ of \folfp-formulas is satisfiable in a \folfp-Henkin model.\label{folfpcomp}
\end{thm}

\begin{proof}
First turn $\Gamma$ into a \folfp maximal consistent set $\Gamma^*$ with \folfp-Henkin witnesses in a possibly richer signature (with extra
individual constants and monadic predicates) $\sigma^*$. Then build a structure $\frM_{\Gamma^*}$ out of this $\Gamma^*$. Then the structure
$\frM_{\Gamma^*}$ satisfies $\Gamma^*$ under the valuation $g_{\Gamma^*}$ and hence it satisfies also $\Gamma$ ($\Gamma$ being a subset of $\Gamma^*$).
\end{proof}

Compactness follows directly from Definition \ref{defax} and Theorem
\ref{hencomp}, i.e., a possibly infinite set of $\Lambda$-sentences
has a $\Lambda$-Henkin model if and only if every finite subset of it
has a $\Lambda$-Henkin model. It also follows directly from Theorem
\ref{hencomp} that $\vdash_\Lambda^{tree}$ is complete on the class of
its $\Lambda$-Henkin-models. Nevertheless, by compactness the axioms
of $\vdash_\Lambda^{tree}$ also have infinite models. We overcome this problem by defining a slightly larger class of Henkin structures, which we will call \emph{definably well-founded $\Lambda$-quasi-trees}.\footnote{For a nice picture of a quasi-tree that is \emph{not} definably well-founded, see \cite{1995}.}
\begin{defi}
  A $\Lambda$-\emph{quasi-tree} is any $\Lambda$-Henkin structure $$(T,<,\prec,P_1,\ldots,P_n,\mathbb{A}_T)$$
  (where $\mathbb{A}_T$ is the set of admissible subsets of $T$)
  satisfying the axioms T1--T10 of
  Figure~\ref{fig:tree-axioms}. A $\Lambda$-quasi-tree is \emph{definably
    well founded} if, in addition, it satisfies all $\Lambda$-instances of the
  induction scheme Ind of Figure~\ref{fig:tree-axioms}.
\end{defi}

With this definition, we obtain from Theorem~\ref{hencomp} the following:

\begin{cor}\label{henktreecomp}
A set of $\Lambda$-formulas is $\vdash_\Lambda^{tree}$-consistent if
and only if it is satisfiable in a definably well-founded $\Lambda$-quasi-tree.
\end{cor}

\section{Operations on Henkin-Structures}
\label{sec4}

Let $\Lambda \in \{\mso,\fotc,\folfp\}$. As noted in Remark \ref{rem}, every finite $\Lambda$-Henkin structure is also a standard structure.
Hence, when working in finite model theory, it is enough to rely on the usual \fo constructions to define operations on structures.
On the other hand, even though our main completeness result concerns finite trees, inside the proof we need to consider infinite ($\Lambda$-Henkin) structures and operations on them. In this context, methods for forming new structures out of existing ones have to be redefined carefully.
We first propose a notion of substructure of a $\Lambda$-Henkin-structure
generated by one of its parametrically definable admissible subsets:
\begin{defi}
[$\Lambda$-substructure]
    Let $\frM=(dom(\frM),Pred,\mathbb{A}_\frM)$ be a $\Lambda$-Henkin-structure (where $Pred$ is the interpretation of the predicates). We
    call $\frM_\fo=(dom(\frM), Pred)$ the relational structure underlying $\frM$. Given a parametrically definable set $A \in \mathbb{A}_\frM$, the
    $\Lambda$-substructure of $\frM$ generated by $A$ is the structure $\frM \upharpoonright A=(\langle A \rangle_{\frM_\fo},\mathbb{A}_{\frM \upharpoonright A})$,
where $\langle A \rangle_{\frM_\fo}$ is the relational substructure of
$\frM_\fo$ generated by $A$ (note that $A$ forms the domain of
$\langle A \rangle_{\frM_\fo}$, as the vocabulary is purely
relational) and $\mathbb{A}_{\frM \upharpoonright A}= \{X \cap A | X \in \mathbb{A}_\frM\}$.
\label{sub}
\end{defi}

Note that in the case of \mso and \folfp, we could also have defined $\mathbb{A}_{\frM \upharpoonright A}$ in an alternative way:

\begin{prop}
Take $\frM$ and $A$ as previously and consider the structure
$(\frM \upharpoonright A)'=(\langle A \rangle_{\frM_\fo},
\mathbb{A}_{(\frM \upharpoonright A)'})$,
where
$\mathbb{A}_{(\frM \upharpoonright A)'}=\{X \subseteq A | X
    \in \mathbb{A}_\frM\}$. Whenever $\frM$ is a \mso-Henkin structure or a \folfp-Henkin structure, $\frM \upharpoonright A$ and $(\frM \upharpoonright A)'$ are one and the same structure.
\label{altsub}
\end{prop}
\begin{proof}
Indeed, take $B \in \mathbb{A}_{\frM \upharpoonright A}$. So there exists $B' \in \mathbb{A}_\frM$ such that $B=B' \cap A$. We want to show that also $B' \cap A \in \mathbb{A}_{(\frM \upharpoonright A)'}$ i.e. $B'\cap A \subseteq A$ (which obviously holds) and $B' \cap A\in \mathbb{A}_\frM$. The second condition holds because both $B'$ and $A$ are parametrically definable in $\frM$, so their intersection also is ($B' \cap A=\{x ~|~ \frM \models Ax \wedge B'x\}$). Conversely, consider $B \in \mathbb{A}_{(\frM \upharpoonright A)'}$. As
$B\subseteq A$ and $B\in \mathbb{A}_\frM$ it follows that $B \in \mathbb{A}_{\frM \upharpoonright A}$ (we can take $B=B \cap A$).
\end{proof}
Now, in order to show that $\Lambda$-substructures are $\Lambda$-Henkin-structures, we introduce a notion of \emph{relativization} and a corresponding \emph{relativization lemma}.  This lemma establishes that for every $\Lambda$-Henkin-structure $\frM$ and $\Lambda$-substructure $\frM \upharpoonright A$ of $\frM$ (with $A$ a set parametrically definable in $\frM$), if a set is parametrically definable in $\frM \upharpoonright A$ then it is also parametrically definable in $\frM$. This result will be useful again in Section \ref{sectproof}.
\begin{defi}[Relativization mapping]
Given two $\Lambda$-formulas $\varphi$, $\psi$ having no variables
in common and given a \fo variable $x$ occurring free in $\psi$,
we define $REL(\varphi,\psi,x)$ by induction on the complexity of
$\varphi$ and call it the \emph{relativization of $\varphi$ to
$\psi$}:
\begin{iteMize}{$\bullet$}
 \item If $\varphi$ is an atom, $REL(\varphi,\psi,x)=\varphi$,
 \item If $\varphi:\approx\varphi_1 \wedge \varphi_2$, $REL(\varphi,\psi,x)=REL(\varphi_1,\psi,x) \wedge REL(\varphi_2,\psi,x)$ (similar for $\vee, \rightarrow, \neg$),
 \item If $\varphi:\approx\exists y \chi$, $REL(\varphi,\psi,x)=\exists y (\psi[y/x] \land REL(\chi,\psi,x))$,
 \item If $\varphi:\approx\exists Y \chi$, $REL(\varphi,\psi,x)=\exists Y (\forall x(Yx \rightarrow \psi) \land REL(\chi,\psi,x))$,
 \item If $\varphi:\approx[TC_{yz}\chi](u,v)$,\\
  $REL(\varphi,\psi,x)=[TC_{yz} REL(\chi,\psi,x) \wedge \psi[y/x] \wedge \psi[z/x]](u,v)$,
 \item If $\varphi:\approx[LFP_{Xy}\chi]z$, $REL(\varphi,\psi,x)= [LFP_{Xy}\chi \wedge\psi[y/x]]z$.
\end{iteMize}
where $\psi[y/x]$ is the formula obtained by replacing in $\psi$ every occurrence of $x$ by $y$ and similarly for $\psi[z/x]$.
\end{defi}

Hence for instance, $REL(\exists y P(y),Q(x),x)=\exists y (P(y) \wedge Q(y))$, which is satisfied in any model $\frM$ of which the submodel induced by $Q$ contains an element satisfying $P$.

\begin{lem}[Relativization lemma]
Let $\frM$ be a $\Lambda$-Henkin-structure, $g$ a valuation on
$\frM$, $\varphi$, $\psi$ $\Lambda$-formulas having no variable in common and $A=\{x ~|~ \frM,
g \models \psi\}$. If $g(y)\in A$ for every variable $y$ occurring
free in $\varphi$ and $g(Y) \in \mathbb{A}_{\frM \upharpoonright A}$ for every
set variable $Y$ occurring free in $\varphi$, then $\frM,g \models
REL(\varphi,\psi,x) \Leftrightarrow \frM \upharpoonright A, g
\models \varphi$. \label{rel}
\end{lem}

\begin{proof}
By induction on the complexity of $\varphi$. Let $g$ be an
assignment satisfying the required conditions. Base case:
$\varphi$ is an atom and $REL(\varphi,\psi,x)=\varphi$. So $\frM,g
\models \varphi \Leftrightarrow \frM \upharpoonright A,g \models
\varphi$ (by hypothesis, $g$ is a suitable assignment for both
models). Inductive hypothesis: the property holds for every
$\varphi$ of complexity at most $n$. Now consider $\varphi$ of
complexity $n+1$.
\begin{iteMize}{$\bullet$}
 \item  $\varphi :\approx\varphi_1 \wedge \varphi_2$ and $REL(\varphi_1 \wedge \varphi_2,\psi,x):\approx REL(\varphi_1,\psi,x)\wedge REL(\varphi_2,\psi,x)$. By induction hypothesis, the property holds for $\varphi_1$ and for $\varphi_2$. By the semantics of $\wedge$, it also holds for $\varphi_1 \wedge \varphi_2$. (Similar for $\vee, \rightarrow, \neg$.)
 \item $\varphi :\approx \exists y \chi$ and $REL(\exists y \chi):\approx\exists y (\psi[y/x] \land REL(\chi,\psi,x))$. By inductive hypothesis, for every node $a \in A$, $\frM,g[a/y] \models REL(\chi,\psi,x) \Leftrightarrow \frM \upharpoonright A, g[a/y] \models \chi$. Hence, by the semantics of $\exists$ and by definition of $A$, $\frM,g \models \exists y (\psi[y/x] \land REL(\chi,\psi,x)) \Leftrightarrow \frM \upharpoonright A, g \models \exists y \chi$.
 \item $\varphi :\approx\exists Y \chi$ and $REL(\exists Y
   \chi,\psi,x)=\exists Y (\forall x(Yx \rightarrow \psi) \land
   REL(\chi,\psi,x))$. As every admissible subset of $\frM
   \upharpoonright A$ is also admissible in $\frM$ (by Proposition
   \ref{altsub}) it follows by inductive hypothesis that for every $B
   \in \mathbb{A}_{\frM\upharpoonright A}$ with $B\subseteq A$, $\frM,g[B/Y] \models REL(\chi,\psi,x) \Leftrightarrow \frM \upharpoonright A, g[B/Y] \models \chi$. Hence, by the semantics of $\exists$ and by definition of $A$, $\frM,g \models \exists Y (\forall x(Yx \rightarrow \psi) \land REL(\chi,\psi,x)) \Leftrightarrow \frM \upharpoonright A, g \models \exists Y \chi$.
\item \sloppy$\varphi :\approx [TC_{yz}\chi](u,v)$ and $REL([TC_{yz}\chi](u,v),\psi,x)=[TC_{yz} REL(\chi,\psi,x) \wedge \psi[y/x] \wedge \psi[z/x]](u,v)$.
By definition of $TC$, the following are equivalent:
\begin{enumerate}[$1.$]
\item  $\frM \upharpoonright A, g \models [TC_{yz}\chi](u,v)$,
\item for all $B \in \mathbb{A}_{\frM \upharpoonright A}$, if $g(u) \in B$ and for all $a, b \in A$, $a\in B$ and $\frM \upharpoonright A, g[a/y,b/z] \models \chi$ implies $b\in B$, then $g(v)\in B$.
\end{enumerate}
By inductive hypothesis, for all $a, b \in A$,\\
 $\frM, g[a/y,b/z] \models REL(\chi,\psi,x) \Leftrightarrow \frM \upharpoonright A, g[a/y,b/z] \models \chi$. Hence $2. \Leftrightarrow 3.$:
\begin{enumerate}
\item [$3.$] for all $B \in \mathbb{A}_{\frM \upharpoonright A}$, if $g(u) \in B$ and for all $a,b\in A$, $a\in B$ and $\frM, g[a/y,b/z] \models REL(\chi,\psi,x)$ implies $b\in B$, then $g(v)\in B$,
\end{enumerate}
By definition of $A$, $3. \Leftrightarrow 4.$:
\begin{enumerate}
\item [$4.$] for all $B \in \mathbb{A}_{\frM \upharpoonright A}$, if $g(u) \in B$ and for all $a,b\in dom(\frM)$, $a\in B$ and $\frM, g[a/y,b/z] \models REL(\chi,\psi,x) \wedge \psi[y/x] \wedge \psi[z/x]$ implies $b\in B$, then $g(v)\in B$,
\end{enumerate}
We claim that $4. \Leftrightarrow 5.$:
\begin{enumerate}
\item [$5.$]for all $C \in \mathbb{A}_{\frM}$, if $g(u) \in C$ and for all $a, b \in dom(\frM)$, $a\in C$ and $\frM, g[a/y,b/z] \models REL(\chi,\psi,x) \wedge \psi[y/x] \wedge \psi[z/x]$ implies $b\in C$, then $g(v)\in C$,
\end{enumerate}
which, by the semantics of $TC$, is equivalent to:
\begin{enumerate}
\item   [$6.$]$\frM, g \models [TC_{yz} REL(\chi,\psi,x) \wedge \psi[y/x] \wedge \psi[z/x]](u,v)$.
\end{enumerate}
It is clear that $5. \Rightarrow 4.$. For the $4. \Rightarrow 5.$ direction, assume $4.$. Take any set $C\in \mathbb{A}_\frM$ such that $g(u) \in C$ and for all $a,b\in dom(\frM)$, $a\in C$ and $\frM, g[a/y,b/z] \models REL(\chi,\psi,x) \wedge \psi[y/x] \wedge \psi[z/x]$ implies $b\in C$. Let $B=A \cap C$. By Definition \ref{sub}, $B \in \mathbb{A}_{\frM \upharpoonright A}$. Now by our assumptions on $g$ and by definition of $A$, $g[a/y,b/z]$ only assigns points in $A$. So as $B = A \cap C$, $g(u) \in B$ and for all $a,b\in dom(\frM)$, $a\in B$ and $\frM, g[a/y,b/z] \models REL(\chi,\psi,x) \wedge \psi[y/x] \wedge \psi[z/x]$ implies $b\in B$. So by $4.$, $g(v) \in B$. As $B \subseteq C$, it follows that $g(v)\in C$.
\item   $\varphi:\approx[LFP_{Xy}\chi]z$ and $REL([LFP_{Xy}\chi]z,\psi,x):\approx [LFP_{Xy}\chi \wedge\psi[y/x]]z$.
By definition of $LFP$, the following are equivalent:
\begin{enumerate}
\item [$1.$] $\frM \upharpoonright A, g \models [LFP_{Xy}\chi]z$,
\item [$2.$]  for all $B \in \mathbb{A}_{\frM \upharpoonright A}$, if for all $a \in A$, $\frM\upharpoonright A, g[a/y,B/X] \models \chi$ implies $a \in B$, then $g(z)\in B$.
\end{enumerate}
By inductive hypothesis, for all $a\in A$, $B \in \mathbb{\frM \upharpoonright A}$, $\frM, g[a/y,B/X] \models REL(\chi,\psi,x) \Leftrightarrow \frM \upharpoonright A, g[a/y,B/X] \models \chi$. Hence $2.$ is equivalent to $3.$:
\begin{enumerate}
\item [$3.$] for all $B \in \mathbb{A}_{\frM \upharpoonright A}$, if for all $a \in A$, $\frM, g[a/y,B/X] \models REL(\chi,\psi,x)$ implies $a \in B$, then $g(z)\in B$,
\end{enumerate}
By definition of $A$, $3. \Leftrightarrow 4.$:
\begin{enumerate}
\item [$4.$] for all $B \in \mathbb{A}_{\frM \upharpoonright A}$, if for all $a \in dom(\frM)$,\\
 $\frM, g[a/y,B/X] \models REL(\chi,\psi,x)\wedge \psi[y/x]$ implies $a \in B$,\\ then $g(z)\in B$,
\end{enumerate}
We claim that $4. \Leftrightarrow 5.$:
\begin{enumerate}
\item [$5.$] for all $C \in \mathbb{A}_\frM$, if for all $a \in dom(\frM)$,\\
 $\frM, g[a/y,C/X] \models REL(\chi,\psi,x) \wedge \psi[y/x]$ implies $a \in C$,\\
  then $g(z)\in C$,
\end{enumerate}
which, by the semantics of $LFP$, is equivalent to:
\begin{enumerate}
\item   [$6.$]$\frM, g \models [LFP_{Xy} REL(\chi,\psi,x) \wedge \psi[y/x]]z$.
\end{enumerate}
It is clear that $5. \Rightarrow 4.$. For the $4. \Rightarrow 5.$ direction, assume $4.$. Take any set $C\in \mathbb{A}_\frM$ such that for all $a \in dom(\frM)$, $\frM, g[a/y,C/X] \models REL(\chi,\psi,x) \wedge \psi[y/x]$ implies $a \in C$.  Let $B=A \cap C$. By Definition \ref{sub}, $B \in \mathbb{A}_{\frM \upharpoonright A}$. Consider $a \in dom(\frM)$ such that $\frM, g[a/y,B/X] \models REL(\chi,\psi,x) \wedge \psi[y/x]$. As $REL(\chi,\psi,x)$ is positive in $X$ and $X$ does not occur in $\psi$, $\frM, g[a/y,C/X] \models REL(\chi,\psi,x) \wedge \psi[y/x]$. Also by hypothesis $a \in C$. Now as $\frM, g[a/y] \models \psi[y/x]$, by definition of $A$, $a \in A$. So $a \in A \cap C$, i.e, $a \in B$
and since we proved it for arbitrary $a \in dom(\frM)$, by $4.$, $g(z) \in B$. As $B \subseteq C$, it follows that $g(z)\in C$.
\end{iteMize}
\end{proof}

\begin{thm}
Let $\frM$ and $A$ be as in Definition \ref{sub}. Then $\frM \upharpoonright A$ is a $\Lambda$-\emph{Henkin}-structure.
\end{thm}

\begin{proof}
Take $B$ parametrically definable in $\frM \upharpoonright A$,
i.e., there is a $\Lambda$-formula $\varphi(y)$ and an assignment
$g$ such that $B=\{a \in dom(\frM \upharpoonright A)~|~\frM
\upharpoonright A,g[a/y]\models \varphi(y)\}$. Now we know that
$A$ is also parametrically definable in $\frM$, i.e., there is a
$\Lambda$-formula $\psi(x)$ and an assignment $g'$ such that
$A=\{a \in dom(\frM) ~|~\frM,g'[a/x]\models \psi(x)\}$. Assume
without loss of generality that $\varphi$ and $\psi$ have no variables in common. We
define an assignment $g^*$ by letting $g^*(z)=g'(z)$ for every
variable $z$ occurring in $\psi$ and $g^*(z)=g(z)$ otherwise. The
situation with set variables is symmetric. Now by Lemma \ref{rel},
$B=\{a \in dom(\frM)~|~\frM,g^*[a/x]\models REL(\varphi,\psi,x)\}$
and hence $B \in \mathbb{A}_\frM$. By definition \ref{sub} it follows that $B \in  \mathbb{A}_{\frM \upharpoonright A}$ (because $B=B\cap A$).
\end{proof}
There is, in model theory, a whole range of methods to form new structures out of
existing ones. Standard references on the matter are \cite{1959,2004}, written in a very general algebraic setting.
Familiar constructions like disjoint unions of relational structures are redefined as particular cases of a new notion of
\emph{generalized product} of \fo-structures and abstract properties of such products are studied. In particular, an important
theorem now called the Feferman-Vaught theorem for \fo is proven in  \cite{1959}. We are particularly interested in one of its corollaries, which
establishes that generalized products of relational structures preserve elementary equivalence. We show an
analogue of this result for a particular case of generalized product of $\Lambda$-Henkin-structures that we call \emph{fusion}, this notion
being itself a generalization of a notion of disjoint union of $\Lambda$-Henkin-structures defined below.
\begin{defi}[Disjoint union of $\Lambda$-Henkin-structures]
  Let $\sigma$ be a purely relational vocabulary and $\sigma^* = \sigma \cup \{Q_1, \ldots, Q_k\}$,
  with $\{Q_1, \ldots, Q_k\}$ a set of new  monadic predicates.
  For any $\Lambda$-Henkin-structures $\frM_1, \ldots, \frM_k$ in vocabulary $\sigma$ with disjoint domains,
  define their \emph{disjoint union} $\biguplus_{ 1 \leq i \leq k} \frM_i$ (or, \emph{direct sum})
  to be the $\sigma^*$-frame that
  has as its domain the union of the domains of the structures $\frM_i$ and likewise for the relations, except for
  the predicates $Q_i$, whose interpretations are
  respectively defined as the domain of the structures $\frM_i$ (we will use $Q_i$ to label the elements of $M_i$).
  The set of admissible subsets $\mathbb{A}_{\biguplus_{1\leq i\leq k}\frM_i}$ is the closure under finite union of the union of the sets of
  admissible subsets of the $\frM_i$. That is:
  \begin{iteMize}{$\bullet$}
    \item   $dom(\biguplus_{ 1 \leq i \leq k} \frM_i)=\bigcup_{ 1 \leq i \leq k} dom(\frM_i)$
    \item   $P^{\biguplus_{ 1 \leq i \leq k}
    \frM_i}=\bigcup_{ 1 \leq i \leq k} P^{\frM_i}$ (with $P \in \sigma$) and $Q_i^{\biguplus_{ 1 \leq i \leq k}
    \frM_i}=dom(\frM_i)$
    \item   $A \in \mathbb{A}_{\biguplus_{1\leq i\leq k}\frM_i}$ iff $A= \bigcup_{1\leq i\leq k} A_i$ for some $A_i
    \in \mathbb{A}_{\frM_i}$
  \end{iteMize}
\label{union}
  \end{defi}

\begin{defi}[$f$-fusion of $\Lambda$-Henkin-structures]
Let $\sigma$ be a purely relational vocabulary and $\sigma^* = \sigma \cup \{Q_1, \ldots, Q_k\}$,
with $\{Q_1, \ldots, Q_k\}$ a set of new  monadic predicates.
Let $f$ be a function mapping each $n$-ary predicate $P \in
\sigma$ to a quantifier-free first-order formula over $\sigma^*$ in variables
$x_1,\ldots,x_n$. For any $\Lambda$-Henkin-structures
$\frM_1, \ldots, \frM_k$ in vocabulary $\sigma$ with disjoint domains, define their
\emph{$f$-fusion} to be the $\sigma$-frame $\bigoplus_{1
\leq i \leq k}^f \frM_i$ that has the same domain and set of
admissible subsets as $\biguplus_{1 \leq i \leq k} \frM_i$.
For every $P\in \sigma$, the interpretation of $P$ in
$\bigoplus_{1 \leq i \leq k}^f \frM_i$ is the set of
$n$-tuples satisfying $f(P)$ in
$\biguplus_{1\leq i\leq k} \frM_i$.
\label{fusion}
\end{defi}

An easy example of $f$-fusion on standard structures (it
is simpler to give an example on standard structures, as we do not have to say anything about admissible sets) is
the ordered sum of two linear orders $(M_1,<_1),(M_2,<_2)$, where
all the elements of $M_1$ are before the elements of $M_2$. In
this case, $\sigma$ consists of a single binary relation $<$, the
elements of $M_1$ are indexed with $Q_1$, those of $M_2$ with
$Q_2$ and $f$ maps $<$ to $x_1< x_2 \vee (Q_1x_1 \wedge Q_2 x_2)$. Another notable example of $f$-fusion is the  $\sigma \cup \{Q_1, \ldots, Q_k\}$-structure
$\biguplus_{1 \leq i \leq k} \frM_i=\bigoplus_{1 \leq i \leq k}^f \frM_i^+$, where $f$ is the identity function and for each $1\leq i \leq k$,
$\frM_i^+$ is the expansion of the $\sigma$-structure $\frM_i$ in which $Q_i^{\frM_i^+}=dom(\frM_i)$ and  $Q_j^{\frM_i^+}=\emptyset$ for every $i \neq j$. In this sense, disjoint
union as we defined it above can be seen as a special case of fusion.

We show preservation results involving $f$-fusions of
$\Lambda$-Henkin-structures. Hence we deal with analogues of
elementary equivalence for these logics and we refer to
\emph{$\Lambda$-equivalence}. Let us recall that by quantifier depth of a $\Lambda$-formula, we mean the maximal number of nested quantifiers in the formula (by ``quantifier'', we mean \fo and \mso-quantifiers, as well as $TC$ or $LFP$-operators).

\begin{defi}
Given two $\Lambda$-Henkin-structures $\frM$ and $\frN$, we write
$\frM\equiv_\Lambda\frN$ and say that $\frM$ and $\frN$ are
\emph{$\Lambda$-equivalent} if they satisfy the same $\Lambda$-sentences. Also, for any natural number $n$, we write
$\frM\equiv_\Lambda^n\frN$ and say that $\frM$ and $\frN$ are \emph{$n$-$\Lambda$-equivalent} if $\frM$ and $\frN$ satisfy the same
$\Lambda$-sentences of quantifier depth at most $n$. In
particular, $\frM\equiv_\Lambda\frN$ holds iff, for all $n$,
$\frM\equiv_\Lambda^n\frN$ holds.
\end{defi}

Now we are ready to introduce the ``Feferman-Vaught theorems'' that we will show in Section \ref{apc} and which establish that $f$-fusions of $\Lambda$-Henkin-structures preserve
$\Lambda$-equivalence, that is:

\begin{thm}
Let $\frM_1,\ldots,\frM_k$, $\frN_1,\ldots,\frN_k$ be $\Lambda$-Henkin structures. Whenever $\frM_i \equiv^n_\Lambda \frN_i$ for all $1 \leq i \leq k$,
then also $\bigoplus_{1 \leq i \leq k}^f \frM_i \equiv^n_\Lambda \bigoplus_{1 \leq i \leq k}^f \frM_i$.
\end{thm}

We will also show in this section that every $f$-fusion of $\Lambda$-Henkin-structures is a $\Lambda$-Henkin-structure.
Comparable work had already been done by Makowski in \cite{2004}
for extensions of \fo, but an important difference is that he only
considered standard structures, whereas we need to deal with
$\Lambda$-Henkin-structures. Our proofs make use of
Ehrenfeucht-Fra\"{\i}ss\'e games for each of the
logics $\Lambda$. 

\subsection{Ehrenfeucht-Fra\"\i ss\'e Games on Henkin-Structures}\
\label{apb}

\noindent
Let $\Lambda \in \{\mso,\fotc,\folfp\}$. We survey Ehrenfeucht-Fra\"{\i}ss\'e games
for \fo, \mso, \fotc, and \folfp which are suitable to use on Henkin structures.
We also provide an adequacy proof for the \fotc game. 
The \mso game is a rather straightforward
extension of the \fo case and has already been used by other
authors (see for instance \cite{1998}). The \folfp game is
borrowed from Uwe Bosse \cite{736408}. It also applies to Henkin
structures, as careful inspection shows. The \fotc game has
already been mentioned in passing by Erich Gr\"{a}del in \cite{736267}
as an alternative to the game he used and we show that it is
adequate for Henkin semantics. It looks also similar to a system
of partial isomorphisms given in \cite{1992}. However it is
important to note that this game is very different from the
\fotc game which is actually used in \cite{736267}. The two games
are equivalent when played on standard structures, but not when
played on \fotc-Henkin structures. This is so because the game
used in \cite{1992} relies on the alternative semantics for the
$TC$ operator given in Proposition \ref{tcaltsem}, so that only
finite sets of points can be chosen by players ; whereas the game
we use involves choices of not necessarily finite admissible
subsets. These are not equivalent approaches. Indeed, on
\fotc-Henkin structures a simple compactness argument shows that
the semantical clause of Proposition \ref{tcaltsem} (defined in
terms of existence of a \emph{finite} path) is not adequate.

Let us first introduce basic notions connected to these games. One
rather trivial sufficient condition for $\Lambda$-equivalence is
the existence of an \emph{isomorphism}.  Clearly isomorphic
structures satisfy the same $\Lambda$-formulas.  A more
interesting sufficient condition for $\Lambda$-equivalence is that
of Duplicator having a winning strategy in all $\Lambda$
Ehrenfeucht-Fra\"{\i}ss\'e games of finite length. To define this,
we first need this notion:

\begin{defi}[Finite Partial Isomorphism]
A \emph{finite partial isomorphism} between structures $\frM$ and
$\frN$ is a finite relation $\{(a_1,b_1),\ldots,(a_n,b_n)\}$
between the domains of $\frM$ and $\frN$ such that for all atomic
formulas $\varphi(x_1, \ldots, x_n)$, $\frM\models\varphi~[a_1,
\ldots, a_n]$ iff $\frN\models\varphi~[b_1, \ldots, b_n]$. Since
equality statements are atomic formulas, every finite partial
isomorphism is (the graph of) a \emph{injective partial function}.
\end{defi}

We will also need the following lemma:

\begin{lem}[Finiteness Lemma]
Fix any set $x_1, \ldots, x_k, X_{k+1}, \ldots, X_m$. In a finite
relational vocabulary, up to logical equivalence, with these free
variables, there are only finitely many $\Lambda$-formulas of quantifier depth $\leq n$.
\label{finitenesslemma}
\end{lem}

\begin{proof}
This can be shown by induction on $k$. In a finite relational vocabulary, with finitely many free variables, there are only finitely many atomic formulas. Now, any $\Lambda$-formula of quantifier depth $k+1$ is equivalent to a Boolean combination of atoms and formulas of quantifier depth $k$ prefixed by a quantifier. Applying a quantifier to equivalent formulas preserves equivalence and the Boolean closure of a finite set of formulas remains finite, up to logical equivalence.
\end{proof}

Now, as we are concerned with extensions of \fo, every $\Lambda$-game will be defined as an extension of the classical \fo game, that we recall here:

\begin{defi}[\fo Ehrenfeucht-Fra\"\i ss\'e Game]
The \fo Ehrenfeucht-Fra\"\i ss\'e game of length $n$ on standard structures $\frM$ and $\frN$ (notation: ${EF_{FO}^n}(\frM,\frN))$ is as follows.
There are two players, Spoiler and Duplicator. The game has $n$ rounds, each of which consists of a move of Spoiler followed by a move of
Duplicator. Spoiler's moves consist of picking an element from one of the two structures, and Duplicator's responses consist of picking an element
in the other structure. In this way, Spoiler and Duplicator build up a finite binary relation between the domains of the two structures: initially, the
relation is empty; each round, it is extended with another pair. The winning conditions are as follows: if at some point of the game the
constructed binary relation is not a finite partial isomorphism, then Spoiler wins immediately. If after each round the relation is a finite
partial isomorphism, then the game is won by Duplicator.
\end{defi}

\begin{thm}[\fo Adequacy]\label{thm:effo}
  Assume a finite relational first-order language. Duplicator has a winning strategy in the game
  $EF_{FO}^n(\frM,\frN)$ iff $\frM\equiv^n_{FO}\frN$. In particular,
  Duplicator has a winning strategy in all EF-games of finite length
  between $\frM$ and $\frN$ if and only if $\frM\equiv_{FO}\frN$.
\end{thm}

The proof for the first order case is classic. We refer the reader
to the proof given in \cite{ebfl95} or to the one in \cite{leo}.

For technical convenience in the course of inductive proofs,
we extend the notion of \fo parameter by considering set
parameters, i.e., instead of interpreting a set variable as a name
of the admissible set $A$, we can add a new monadic predicate $A$ to the
signature. The new predicates and the sets they name are called
set parameters. (This is similar to the \fo notion that can be
found in \cite{97}.) We will work with \emph{parametrized} (or \emph{expanded}) Henkin-structures, that is, structures considered together with partial
valuations. This means that the assignment is possibly non empty at the beginning of the game, which can start with some ``handicap'' for Duplicator, i.e., some preliminary set of already ``distinguished objects and sets''.

We first define a necessary and sufficient condition for \mso
equivalence by extending Ehrenfeucht-Fra\"{\i}ss\'e games from
\fo to \mso. This game has already been defined in the literature, see for instance \cite{1998}.

\begin{defi}[\mso Ehrenfeucht-Fra\"\i ss\'e Game]
Consider two \mso-Henkin structures $\frM$ together with $\bar{A} \in \mathbb{A}_{\frM}^r$,
$\bar{a} \in dom(\frM)^s$ and $\frN$ together with $\bar{B} \in
\mathbb{A}_{\frN}^r$, $\bar{b} \in dom(\frN)^s$ and $r \geq 0$, $s \geq 0$, $n \geq 0$. The \mso
Ehrenfeucht-Fra\"{\i}ss\'e game $EF_\mso^n((\frM, \bar{A}, \bar{a}),
(\frN, \bar{B}, \bar{b}))$ of length $n$ on expanded structures
$(\frM, \bar{A}, \bar{a})$ and $(\frN, \bar{B}, \bar{b})$ is
defined as for the first-order case, except that each time she
chooses a structure, Spoiler can choose either an element or an
admissible subset of its domain. For a given $A_{r+1} \in
\mathbb{A}_{\frM}$ chosen by Spoiler, $(\frM, \bar{A}, \bar{a})$
is expanded to $(\frM, \bar{A}, A_{r+1}, \bar{a})$. Duplicator
then responds by choosing $B_{r+1} \in \mathbb{A}_{\frN}$ and
$(\frN, \bar{B}, \bar{b})$ is expanded to $(\frN, \bar{B},
B_{r+1}, \bar{b})$. The game goes on with the so expanded structures.
The winning conditions are as follows: if at some point of the
game $\bar{a} \mapsto \bar{b}$ is not a finite partial isomorphism
from $(\frM, \bar{A}, A_{r+1})$ to $(\frN, \bar{B}, B_{r+1})$,
then Spoiler wins immediately. If after each round the relation is
a finite partial isomorphism, then the game is won by
Duplicator.
\label{efmso}
\end{defi}

\begin{thm}[\mso Adequacy]\label{thm:efmso}
 Assume a finite relational \mso language. Given $\frM$ and $\frN$, $\bar{A} \in \mathbb{A}_{\frM}^r$, $\bar{B} \in
 \mathbb{A}_{\frN}^r$, $\bar{a} \in dom(\frM)^s$, $\bar{b} \in dom(\frN)^s$ and  $r \geq 0$, $s \geq 0$, $n \geq 0$,
 Duplicator has a winning strategy in the game $EF_\mso^n((\frM,
 \bar{A}, \bar{a}), (\frN, \bar{B}, \bar{b}))$ iff
 $(\frM, \bar{A}, \bar{a})$ and $(\frN, \bar{B}, \bar{b})$ satisfy
 the same \mso formulas of quantifier depth $n$. In
 particular, Duplicator has a winning strategy in all $EF_\mso$-games of finite length between $(\frM, \bar{A}, \bar{a})$
 and $(\frN, \bar{B}, \bar{b})$ if and only if $(\frM, \bar{A}, \bar{a})$ and $(\frN, \bar{B}, \bar{b})$ satisfy the same \mso formulas.
\end{thm}

We omit the proof, because it parallels the \fo case.
The proof works regardless whether \mso is interpreted in the standard or in the Henkin way. What matters here is that the game-theoretic meaning of a ``quantification'' over a given ``domain'', lies in the choice of an element from that domain (including one consisting of ``higher-order elements'', e.g., sets).

\begin{cor}
For \mso-Henkin-structures $\frM$, $\frN$ and $n \geq 0$, Duplicator has a winning strategy in $EF_\mso^n(\frM,\frN)$ if and only if $\frM \equiv_\mso^n
\frN$. In particular, Duplicator has a winning strategy in all $EF_\mso$-games of finite length between $\frM$ and $\frN$ if and only if $\frM
\equiv_\mso \frN$.
\end{cor}

The \fotc game that we will be introducing now had been already mentioned in passing by Erich Gr\"{a}del in \cite{736267} as an alternative to the game he used. We will show that it is adequate on Henkin-structures.

\begin{defi}[\fotc Ehrenfeucht-Fra\"\i ss\'e Game]
Consider two \fotc-Henkin structures $\frM$ and $\frN$ together with
$\bar{a} \in dom(\frM)^s$, $\bar{b} \in dom(\frN)^s$ and $s \geq 0$, $n \geq 0$. The \fotc-game $EF_\fotc^n((\frM, \bar{a}),
(\frN, \bar{b}))$ of length $n$ on expanded structures
$(\frM, \bar{a})$ and $(\frN, \bar{b})$ is
defined as for the first-order case, except that each time she
chooses a structure, Spoiler can either choose only one element or an admissible subset together with two elements of its domain.
In the first case we say that she plays an $\exists$ (or point) move and in the second case, a $TC$-move (which we will define more precisely below).
Each point move results in an extension of the assignment $\{\bar{a}\mapsto \bar{b}\}$ with elements $a_{s+1} \in dom(\frM), b_{s+1} \in dom(\frN)$. Each $TC$-move results in an extension of the assignment $\{\bar{a}\mapsto
\bar{b}\}$ with elements $a_{s+1}, a_{s+2} \in dom(\frM), b_{s+1}, b_{s+2} \in dom(\frN)$. At each round, Spoiler chooses the kind of move to be played. 

The $\exists$ move is defined as in the \fo case. The $TC$-move is as follows:

Spoiler considers two pebbles $(a_i,b_i)$ and $(a_j,b_j)$ on the board (i.e., corresponding couples of parameters taken in each structure) and depending on the structure that he chooses to consider, he plays:
\begin{iteMize}{$\bullet$}
\item   either a set $A \in \mathbb{A}_\frM$ with $a_i \in A$ and $a_j \notin
A$. Duplicator then answers with a set $B \in \mathbb{A}_\frN$ such that
$b_i \in B$ and $b_j \notin B$.  Spoiler now picks $b_{s+1} \in B,
b_{s+2} \notin B$ and Duplicator answers with $a_{s+1} \in A, a_{s+2}
\notin A$.
\item   or a set $B \in \mathbb{A}_\frN$ with $b_i \in B$ and $b_j \notin
B$.  Duplicator then answers with a set $A \in \mathbb{A}_\frM$ such
that $a_i \in A$ and $a_j \notin A$. Spoiler now picks $a_{s+1} \in A,
a_{s+2} \notin A$ and Duplicator answers with $b_{s+1} \in B, b_{s+2}
\notin B$.
\end{iteMize}
In each $TC$-move, the assignment is extended with $a_{s+1} \mapsto
b_{s+1}, a_{s+2} \mapsto b_{s+2}$. After $n$ moves, Duplicator has won
if the constructed assignment $\bar{a} \mapsto \bar{b}$ is a
partial isomorphism (i.e. the game continues with the two new
pebbles in each structure, but the sets $A$ and $B$ are
forgotten).
\label{effotc}
\end{defi}

\begin{thm}[\fotc Adequacy]\label{thm:effotc}
  Assume a finite relational \fotc language. Given two \fotc-Henkin structures $\frM$ and
  $\frN$, $\bar{a} \in dom(\frM)^s$, $\bar{b} \in dom(\frN)^s$ and $r \geq 0$, $s \geq 0$, $n \geq 0$,
  Spoiler has a winning strategy in the game $EF_\fotc^n((\frM, \bar{a}), (\frN, \bar{b}))$ iff
  there is a \fotc formula of quantifier depth $n$ distinguishing $(\frM, \bar{a})$ and $(\frN, \bar{b})$.
\end{thm}

\proof\hfill

\begin{iteMize}{$\Rightarrow$}

\item From the existence of a winning strategy for
Spoiler in the \fotc-game of length $n$ in between $(\frM, \bar{a})$ and $(\frN, \bar{b})$, we will infer the
existence of a \fotc-formula of quantifier depth $n$
distinguishing $(\frM,\bar{a})$ and $(\frN,\bar{b})$.

By induction on $n$.

Base step: With $0$ round the initial match between distinguished objects must have failed to be a partial isomorphism for Spoiler to win. This
implies that $(\frM, \bar{a})$ and $(\frN, \bar{b})$ disagree on some atomic formula.

Inductive step: The induction hypothesis says that for every two
structures, if Spoiler can win their comparison game over $n$ rounds,
then the structures disagree on some \fotc-formula of quantifier depth
$n$.  Assume that for some structures $(\frM, \bar{a}), (\frN,
\bar{b})$, Spoiler has a winning strategy for the game over $n+1$
rounds. Let us reason on Spoiler's first move in the game. It can
either be a $TC$ or an $\exists$ move.

If it is an $\exists$ move, then it means that Spoiler picks an
element $a$ in one of the two structures, so that no matter what
element $b$ Duplicator picks in the other, Spoiler has an
$n$-round winning strategy. But then we can use the induction
hypothesis, and find for each such $b$ a formula $\varphi_b(x)$
that distinguishes $(\frM, \bar{a}, a)$ from $(\frN, \bar{b}, b)$.
In fact we can assume that in each case the respective formula is
true of $(\frM, \bar{a}, a)$ and false of $(\frN, \bar{b}, b)$ (by
negating the formula if needed). Now take the big conjunction
$\varphi(x)$ of all these formulas (which is equivalent to a
finite formula according to Lemma \ref{finitenesslemma}) and prefix it
with an existential quantifier. Then the resulting formula is true
in $(\frM, \bar{a})$ but false in $(\frN, \bar{b})$. It is true in
$(\frM, \bar{a})$ if we pick $a$ for the existentially quantified
variable. And no matter which element we pick in $(\frN,
\bar{b})$, it will always falsify one of the conjuncts in the
formula, by construction. So, the new formula is false in $(\frN,
\bar{b})$. I.e., $\exists x \varphi(x)$ of quantifier depth $n+1$
distinguishes $(\frM, \bar{a})$ and $(\frN, \bar{b})$.

If Spoiler's first move is a $TC$-move, then it means that
Spoiler picks a subset in one structure, let say $A \in
\mathbb{A}_\frM$ (with $a_i \in A$ and $a_j \not \in A$), so that
no matter which $B \in \mathbb{A}_\frN$ (with $b_i \in B$ and $b_j
\not \in B$) Duplicator picks in the other structure, Spoiler can
pick $b_k \in B$, $b_{k+1} \not \in B$ such that no matter which
$a_k \in A$, $a_{k+1} \not \in A$ Duplicator picks, Spoiler has an
$n$-round winning strategy. For each $B$ that might be chosen by
Duplicator, Spoiler's given strategy gives a fixed couple
$b_k,b_{k+1}$. For each response $a_k,a_{k+1}$ of Duplicator, we
thus obtain by inductive hypothesis a discriminating formula
$\varphi_{B,a_k,a_{k+1}}(x,y)$ that we can assume to be true in
$(\frN, \bar{b})$ for $b_k,b_{k+1}$ and false in $(\frM, \bar{a})$
for $a_k,a_{k+1}$. Now for each $B$, let us take the big
conjunction $\Phi_B(x,y)$ of all these formulas (which is finite,
by Lemma \ref{finitenesslemma}). We can then construct the big
disjunction $\Phi(x,y)$ (again finite, by the same lemma) of all
the formulas $\Phi_B(x,y)$.

Considering the first round in the game together with the inductive hypothesis, note that the \mso formula $\exists X
(a_i \in X \wedge a_j \not \in X \wedge \forall x y ((x \in X \wedge y \not \in X) \rightarrow \neg\Phi(x,y)))$ holds in $(\frM, \bar{a})$. Indeed, by induction
hypothesis, any couple $a_k \in A, a_{k+1} \not \in A$ that Duplicator might choose in $dom(\frM)$ will always falsify at least one of the conjuncts of
each $\Phi_B(x,y)$. Finally, the formula $\Phi(x,y)$ being constructed as the disjunction of all the formulas $\Phi_B(x,y)$, any such couple
$a_k,a_{k+1}$ will also falsify $\Phi(x,y)$. Now the \mso formula $\exists X (a_i \in X \wedge a_j \not \in X \wedge \forall
x y ((x \in X \wedge y \not \in X) \rightarrow \neg\Phi(x,y)))$ is equivalent to $\exists X (a_i \in X \wedge a_j \not \in X \wedge \neg \exists
x y (x \in X \wedge \Phi(x,y) \wedge y \not \in X))$, which means that $(\frM, \bar{a}) \not \models [TC_{xy}\Phi(x,y)](a_i,a_j)$.

On the other hand for the same reasons, note that it holds in $(\frN, \bar{b})$ that $\forall X ((b_i \in X \wedge b_j \not \in X) \rightarrow
\exists x y (x \in X \wedge y \not \in X \wedge \Phi(x,y)))$. Indeed, by induction hypothesis, for each $B$ that Duplicator might choose in
$\mathbb{A}_{\frN}$ Spoiler will always be able to find a couple $b_k \in B, b_{k+1} \not \in B$ satisfying all the conjuncts of the
corresponding formulas $\Phi_B(x,y)$. Finally, the formula $\Phi(x,y)$ being constructed as the disjunction of all the formulas $\Phi_B(x,y)$, such a couple
$a_k,a_{k+1}$ will also satisfy $\Phi(x,y)$. Now $\forall X ((b_i \in X \wedge b_j \not \in X) \rightarrow \exists x y (x \in X \wedge y \not \in X \wedge \Phi(x,y)))$ is equivalent to $\forall X (b_i \not \in X \vee b_j \in X \vee \exists x y
(x \in X \wedge y \not \in X \wedge \Phi(x,y)))$, which means that $(\frN, \bar{b}) \models [TC_{xy}\Phi(x,y)](b_i,b_j)$.

Let $u$ be a name for the parameters $a_i, b_i$ and $v$ for $b_i, b_j$. $[TC_{xy}\Phi(x,y)](u,v)$ of quantifier depth $n+1$ distinguishes
$(\frN, \bar{a})$ and $(\frM, \bar{b})$.

\item[$\Leftarrow$]From the existence of a \fotc formula of
quantifier depth $n$ distinguishing $(\frM, \bar{a})$ and $(\frN,
\bar{b})$ we will infer the existence of a winning strategy for Spoiler in
the game $EF_{FO+TC}^n((\frM, \bar{a}), (\frN, \bar{b}))$.

By induction on $n$. 

Base step: Doing nothing is a strategy for Spoiler.

Inductive step: The inductive hypothesis says that, for every two
structures, if they disagree on some \fotc formula of quantifier
depth $n$, then Spoiler has a winning strategy in the $n$-round
game. Now, assume that some expanded structures $(\frM, \bar{a})$,
$(\frN, \bar{b})$ disagree on some \fotc formula $\chi$ of
quantifier depth $n+1$. Any such formula must be equivalent to a
Boolean combination of formulas of the form $\exists x \psi(x)$
and $[TC_{xy} \varphi(x,y)](u,v)$ with $\psi$, $\varphi$ of
quantifier depth at most $n$. If $\chi$ distinguishes the two
structures, then there is at least one component of this Boolean
combination which suffices for distinguishing them.

Let us first suppose that it is of the form $\exists x \psi(x)$. We may assume without loss of generality that
$(\frM, \bar{a}) \models \exists x \psi(x)$ whereas
$(\frN, \bar{b}) \not \models \exists x \psi(x)$. Then it means
that there exists an object $a \in dom(\frM)$ such that $(\frM, \bar{a})
\models \psi(a)$ whereas for every object $b \in dom(\frN)$, $(\frN,
\bar{b}) \not \models \psi(b)$. But then we can use our induction
hypothesis and find for each such $b$ a winning strategy for
Spoiler in $EF_\fotc^n((\frM, \bar{a}, a),(\frN, \bar{b}, b))$. We
can infer that Spoiler has a winning strategy in $EF_\fotc^{n+1}((\frM,
\bar{a}), (\frN, \bar{b}))$. His first move consists in
picking the object $a$ in $dom(\frM)$ and for each response $b$ in $dom(\frN)$ of
Duplicator, the remaining of his winning strategy is the same as in
$EF_\fotc^n((\frM, \bar{a}, a),(\frN, \bar{b}, b))$.

We now suppose that $[TC_{xy} \varphi(x,y)](u,v)$ of
quantifier depth $n+1$ distinguishes the two structures. We may assume without loss of generality that
$(\frM, \bar{a}) \models[TC_{xy} \varphi(x,y)](u,v)$ i.e. it holds
in $(\frM, \bar{a})$ that $\forall X ((a_i \in X \wedge a_j \not
\in X) \rightarrow \exists x y (x \in X \wedge y \not \in X \wedge
\varphi(x,y)))$, whereas $(\frN, \bar{b}) \not \models[TC_{xy}
\varphi(x,y)](u,v)$ i.e. it holds in $(\frN, \bar{b})$ that
$\exists X (b_i \in X \wedge b_j \not \in X \wedge \neg \exists x
y (x \in X \wedge \varphi(x, y) \wedge y \not \in X))$. We want to
show that Spoiler has a winning strategy  in
$EF_\fotc^{n+1}((\frM, \bar{a}), (\frN, \bar{b}))$. Let us
describe her first move. She first chooses $(\frN, \bar{b})$ and
$B \in \mathbb{A}_\frN$ such that $b_i \in B \wedge b_j \not \in B
\wedge \neg \exists x y (x \in B \wedge \varphi(x,y) \wedge y \not
\in B)$. By definition of $TC$, such a set exists. Duplicator has
to respond by picking a set $A$ in $\mathbb{A}_\frM$ containing $a_i$ and not $a_j$. Spoiler then
picks $a_k \in A$ and $a_{k+1} \not \in A$ such that $(\frM,
\bar{a}) \models \varphi(a_k,a_{k+1})$. This is possible because
by definition of $TC$, for any possible choice $A$ of Duplicator (i.e., any set $A$ containing $a_i$ and not $a_j$)
we have $\exists x y (x \in A \wedge y \not \in A \wedge
\varphi(x,y))$. But that means that Duplicator is now stuck and
has to pick $b_k \in B$ and $b_{k+1} \not \in B$ such that $(\frN,
\bar{b}) \not \models \varphi(b_k,b_{k+1})$. Consequently, we have
$(\frN, \bar{b}, b_k, b_{k+1}) \not \models \varphi(x, y)$,
whereas $(\frM, \bar{a}, a_k, a_{k+1}) \models \varphi(x, y)$. As
$\varphi(x,y)$ is of quantifier depth $n$, by induction
hypothesis, Spoiler has a winning strategy in $EF_\fotc^n((\frM,
\bar{a}, a_k, a_{k+1}),(\frN, \bar{b}, b_k, b_{k+1}))$. The
remaining of Spoiler's winning strategy in $EF_\fotc^{n+1}((\frM,
\bar{a}),(\frN, \bar{b}))$ (i.e. after her first move, that we
already accounted for) is consequently as in $EF_\fotc^n((\frM,
\bar{a}, a_k, a_{k+1}),(\frN, \bar{b}, b_k, b_{k+1}))$.\qed
\end{iteMize}

\begin{cor}
For structures $\frM$, $\frN$ and $n \geq 0$, Duplicator has a winning strategy in $EF_\fotc^n(\frM,\frN)$ if and only if $\frM \equiv_\fotc^n
\frN$. In particular, Duplicator has a winning strategy in all $EF_\fotc$-games of finite length between $\frM$ and $\frN$ if and only if $\frM
\equiv_\fotc \frN$.
\end{cor}

Let us finally consider the \folfp case. There are two classical
equivalent syntactic ways of defining the syntax of \folfp: the one
we used in Section \ref{seclog} and another one, dispensing with
restrictions to positive formulas, but allowing negations only in
front of atomic formulas and introducing a greatest fixed-point
operator as the dual of the least fixed-point operator (also
$\forall$ cannot be defined using $\exists$ and has to be
introduced separately, similarly for the Boolean connectives).
This second way of defining \folfp turns out to be more convenient
to define an adequate Ehrenfeucht-Fra\"{\i}ss\'e game. The game is
suitable to use on Henkin structures because the semantics on
which it relies is merely a syntactical variant of the one given
in Section \ref{sec3}. Now the \folfp-formulas $[LFP_{Xx}
\varphi(x,X)]y$ and $[GFP_{Xx} \varphi(x,X)]y$, stating that a
point belongs to the least fixed-point, or respectively, to the
greatest fixed-point induced by the formula $\varphi$ satisfy the
following equations:
\begin{center}
$[LFP_{Xx} \varphi(x,X)]y \leftrightarrow \forall X (\neg Xy
\rightarrow \exists x (\neg Xx \wedge \varphi(x,X)))$

$[GFP_{Xx} \varphi(x,X)]y \leftrightarrow \exists X (Xy \wedge
\forall x(Xx \rightarrow \varphi(x,X)))$
\end{center}
Note that this holds no matter whether we be concerned with \folfp and \mso on
standard structures or on Henkin structures. The consideration of
these equations is the key idea behind an
Ehrenfeucht-Fra\"{\i}ss\'e game defined by Uwe Bosse in
\cite{736408} for least fixed-point logic \lfp (i.e. where
fixed-points are not only considered for monadic operators, but
for any $n$-ary operator). \folfp being simply the monadic
fragment of \lfp, the game for \lfp can be adapted to \folfp in a
straightforward way:

\begin{defi}[\folfp Ehrenfeucht-Fra\"\i ss\'e game]
Consider \folfp-Henkin structures $\frM$ and $\frN$ together with $\bar{a} \in dom(\frM)^s$,
$\bar{b} \in dom(\frN)^s$,  $\bar{A} \in \mathbb{A}_\frM^r$,
$\bar{b} \in \mathbb{A}_\frN^r$, $r \geq 0$, $s \geq 0$, $n \geq 0$. In the game
$EF_\folfp^n((\frM, \bar{A}, \bar{a}),(\frN, \bar{B}, \bar{b}))$ of length $n$,
there are two types of moves, point and fixed-point moves. Each
move results in an extension of the assignment $\bar{a} \mapsto \bar{b}, \bar{A}
\mapsto \bar{B}$ with elements $a_{s+1} \in dom(\frM), b_{s+1} \in
dom(\frN)$, and possibly (in the case of fixed-point moves) with
sets $A_{r+1} \in \mathbb{A}_\frM, B_{r+1} \in \mathbb{A}_\frN$. Spoiler chooses the kind of move to be played. 
Now the following moves are
possible:

\begin{iteMize}{$\bullet$}
\item{$\exists$ move: Spoiler chooses $a_{s+1} \in dom(\frM)$ and Duplicator $b_{s+1} \in dom(\frN)$.}
\item{$\forall$ move: Spoiler chooses $b_{s+1} \in dom(\frN)$ and Duplicator $a_{s+1} \in dom(\frM)$.}
\end{iteMize}
In each point move, the assignment is extended by $a_{s+1}\mapsto
b_{s+1}$.
\begin{iteMize}{$\bullet$}
\item{$LFP$ move: Spoiler chooses $B_{r+1} \in \mathbb{A}_{\frN} \setminus \{dom(\frN)\}$ with some pebble $b_i \not \in B_{r+1}$
and Duplicator responds with $A_{r+1} \in
\mathbb{A}_{\frM}\setminus \{dom(\frM)\}$.

Now Spoiler chooses in $dom(\frM)$ a new element $a_{s+1} \not \in
A_{r+1}$ and Duplicator answers in $dom(\frN)$ with $b_{s+1} \not \in
B_{r+1}.$}
\item{$GFP$ move:  Spoiler chooses $A_{r+1} \in \mathbb{A}_{\frM}\setminus \{dom(\frM)\}$ with some pebble $a_i \in A_{r+1}$
and Duplicator responds with $B_{r+1} \in
\mathbb{A}_{\frN}\setminus \{dom(\frN)\}$ such that $B_{r+1} \neq
\emptyset$.

Now Spoiler chooses in $dom(\frN)$ a new element $b_{s+1} \in B_{r+1}$ and
Duplicator answers in $dom(\frM)$ with $a_{s+1} \in A_{r+1}$.}
\end{iteMize}
In each fixed-point move the assignment is extended by $A_{r+1}
\mapsto B_{r+1}, a_{s+1} \mapsto b_{s+1}.$

After $n$ moves, Duplicator has won if the constructed element
assignment $\bar{a} \mapsto \bar{b}$ is a partial isomorphism and
for the subset assignment $\bar{A} \mapsto \bar{B}$, for any $1
\leq j \leq r$ and $i \leq s$:
\begin{center}
$a_i \in A_j$ implies $b_i \in B_j$
\end{center}
We call an assignment with these properties a \emph{posimorphism}.
\label{effolfp}
\end{defi}

\begin{thm}[\folfp Adequacy]\label{thm:effolfp}
  Assume a finite relational \folfp language. Given two \folfp-Henkin
  structures $\frM$ and $\frN$, $\bar{A} \in \mathbb{A}_{\frM}^r$,
  $\bar{B} \in \mathbb{B}_{\frN}^r$, $\bar{a} \in dom(\frM)^s$,
  $\bar{b} \in dom(\frN)^s$ and $r \geq 0$, $s \geq 0$, $n \geq 0$,
  Duplicator has a winning strategy in the game $EF_\folfp^n((\frM,
  \bar{A}, \bar{a}), (\frN, \bar{B}, \bar{b}))$ iff $(\frM, \bar{A},
  \bar{a})$ and $(\frN, \bar{B}, \bar{b})$ satisfy the same
  \folfp-for\-mulas of quantifier depth $n$.
\end{thm}
For a proof in the case of standard structures, we refer the reader to Uwe Bosse \cite{736408}. As pointed out earlier, the same argument works as well in the case of Henkin structures.

\subsection{Fusion Theorems on Henkin-Structures}
\label{apc}

Let $\Lambda \in \{\mso,\fotc,\folfp\}$. We show our analogues of Feferman-Vaught theorem for fusions of $\Lambda$-Henkin-structures. We will refer to them as \emph{$\Lambda$-fusion Theorems}, even though they will sometimes be formally first stated as corollaries. What we show is, more precisely, that fusion of $\Lambda$-Henkin-structures preserve $\Lambda$-equivalence for all fixed quantifier-depths.

In order to give inductive proofs for \mso and \folfp, it will be more convenient to consider parametrized $\Lambda$-Henkin-structures where the set of set parameters is closed under union, this notion being defined below. This is safe because whenever two parametrized structures $(\frM,\bar{A},\bar{a})$ and $(\frN,\bar{B},\bar{b})$ are $n$-$\Lambda$-equivalent, it follows trivially that $\frM$ and $\frN$ considered together with a subset of this set of parameters are also $n$-$\Lambda$-equivalent.

\begin{defi}
Let $A_1, \ldots, A_k$ be a finite sequence of set parameters. We
define the sequence $(A_1, \ldots, A_k)^\cup$ as the finite sequence of set
parameters obtained by closing the set $\{A_1, \ldots, A_k\}$
under union, i.e., $(A_1, \ldots, A_k)^\cup=\{\bigcup_{i \in I}
A_i | I \subseteq \{1,\ldots,k\}\}$. (We additionally assume that
this set is ordered in a fixed canonical way, depending on the
index sets $I$.)
\end{defi}

\begin{thm}[Fusion Theorem for \mso]
Let $\mathfrak{M}_i$ and $\mathfrak{N}_i$ be $\mso$-Henkin
structures, where $1\leq i\leq k$. Furthermore, for $1\leq i\leq k$, 
let $\bar{a_i}$, $\bar{b_i}$ be sequences of first-order
parameters of the form $a_{i_1},\ldots,a_{i_m}$,
$b_{i_1},\ldots,b_{i_m}$ (where $m\in\mathbb{N}$ may depend on $i$)
and 
$\bar{A_i}$, $\bar{B_i}$ sequences of set parameters of
the form $A_{i_1},\ldots,A_{i_{m'}}$, $B_{i_1},\ldots,B_{i_{m'}}$
(where $m'\in\mathbb{N}$ may again depend on $i$).
Whenever $$(\frM_i,\bar{A_i},
\bar{a_i})\equiv_{MSO}^n(\frN_i,\bar{B_i},\bar{b_i})\text{ for all }1
\leq i \leq k,$$ then also
$$\bigoplus_{1\leq i\leq k}^f \frM_i,(\bar{A_1},
\ldots, \bar{A_k})^\cup, \bar{a_1},\ldots,\bar{a_k}\equiv_{MSO}^n
\bigoplus_{1\leq i\leq k}^f \frN_i, (\bar{B_1},
\ldots, \bar{B_k})^\cup, \bar{b_1},\ldots,\bar{b_k}.$$.
\label{msofusion}
\end{thm}

\proof
We define a winning strategy for Duplicator in the game
$$EF_\mso^n((\bigoplus_{1\leq i\leq k}^f
\frM_i,(\bar{A_1}, \ldots, \bar{A_k})^\cup,
\bar{a_1},\ldots,\bar{a_k}), (\bigoplus_{1\leq i\leq k}^f
\frN_i, (\bar{B_1}, \ldots, \bar{B_k})^\cup,
\bar{b_1},\ldots,\bar{b_k}))$$ out of her winning strategies in
the games $EF_\mso^n((\frM_i,\bar{A_i},
\bar{a_i}),(\frN_i,\bar{B_i},\bar{b_i}))$ by induction on $n$.

Base step: $n=0$, doing nothing is a strategy for Duplicator. We
need to show that $$(\bigoplus_{1\leq i\leq k}^f
\frM_i,(\bar{A_1}, \ldots, \bar{A_k})^\cup,
\bar{a_1},\ldots,\bar{a_k})$$ and $$(\bigoplus_{ 1 \leq i
\leq k}^f \frN_i, (\bar{B_1},\ldots,\bar{B_k})^\cup,
\bar{b_1},\ldots,\bar{b_k})$$ agree on all atomic formulas. Now in
the fusion structures, each atomic formula is defined by $f$ in terms
of a $\sigma^*$-quantifier free formula that is evaluated in the
corresponding disjoint union structure. So it is enough to show that
the disjoint union structures agree on all atomic $\sigma^*$-formulas
and on their Boolean combinations. The initial match between the
distinguished objects in $(\frM_i, \bar{A_i},\bar{a_i})$ and
$(\frN_i, \bar{B_i},\bar{b_i})$ is a partial isomorphism for every
$1\leq i \leq k$, so it is also one for $\biguplus_{1 \leq
i \leq k} \frM_i, \bar{a_1},\ldots,\bar{a_k}$ and
$\biguplus_{1\leq i\leq k} \frN_i,
\bar{b_1},\ldots,\bar{b_k}$ i.e. the two disjoint union structures
extended with \fo parameters agree on all $\sigma^*$-atomic
formulas. We still need to show that it is also one for
$\biguplus_{1 \leq i \leq k} \frM_i, (\bar{A_1},
\ldots, \bar{A_k})^\cup,\bar{a_1},\ldots,\bar{a_k}$ and
$\biguplus_{1\leq i\leq k} \frN_i,(\bar{B_1},
\ldots, \bar{B_k})^\cup, \bar{b_1},\ldots,\bar{b_k}$ i.e. the two
disjoint union structures extended with \fo parameters and the
closure under union of set parameters agree on all
$\sigma^*$-atomic formulas. It is enough to point that for every
parameter $a_{i_j}$, for every $I \subseteq
\{i_1,\ldots,i_{m'},\ldots,k_1,k_{m'}\}$ by construction of
$\bigcup_{i \in I} A_i$ in $(\bar{A_1}, \ldots,
\bar{A_k})^\cup$, the following are equivalent:
\begin{iteMize}{$\bullet$}
\item $\biguplus_{1\leq i\leq k}
\frM_i,(\bar{A_1}, \ldots, \bar{A_k})^\cup,
\bar{a_1},\ldots,\bar{a_k} \models \bigcup_{i \in I} A_ia_{i_j}$,
\item   $\biguplus_{1\leq i\leq k} \frM_i,(\bar{A_1},
\ldots, \bar{A_k})^\cup, A_{i_l}, \bar{a_1},\ldots,\bar{a_k} \models
A_{i_l}a_{i_j}$ for some $i_l$ in $I$.
\end{iteMize}
Similarly for every parameter $b_{i_j}$, by construction of
$\bigcup_{i \in I} B_i$ in $(\bar{B_1}, \ldots,
\bar{B_k})^\cup$, the following are equivalent:
\begin{iteMize}{$\bullet$}
\item $\biguplus_{1\leq i\leq k}
\frN_i,(\bar{B_1}, \ldots, \bar{B_k})^\cup,
\bar{b_1},\ldots,\bar{b_k} \models \bigcup_{i \in I} B_ib_{i_j}$,
\item   $\biguplus_{1\leq i\leq k} \frN_i,(\bar{B_1},
\ldots, \bar{B_k})^\cup, B_{i_l}, \bar{b_1},\ldots,\bar{b_k} \models
B_{i_l}b_{i_j}$ for some $i_l$ in $I$.
\end{iteMize}
But by Duplicator's winning strategy in the small structure games, we
know that the following are equivalent:
\begin{iteMize}{$\bullet$}
\item   $\biguplus_{1\leq i\leq k} \frM_i,(\bar{A_1},
\ldots, \bar{A_k})^\cup, A_{i_l}, \bar{a_1},\ldots,\bar{a_k} \models
A_{i_l}a_{i_j}$ for some $i_l$ in $I$.
\item   $\biguplus_{1\leq i\leq k} \frN_i,(\bar{B_1},
\ldots, \bar{B_k})^\cup, B_{i_l}, \bar{b_1},\ldots,\bar{b_k} \models
B_{i_l}b_{i_j}$ for some $i_l$ in $I$.
\end{iteMize}
So the following are also equivalent:
\begin{iteMize}{$\bullet$}
\item $\biguplus_{1\leq i\leq k}
\frM_i,(\bar{A_1}, \ldots, \bar{A_k})^\cup,
\bar{a_1},\ldots,\bar{a_k} \models \bigcup_{i \in I} A_ia_{i_j}$,
\item $\biguplus_{1\leq i\leq k}
\frN_i,(\bar{B_1}, \ldots, \bar{B_k})^\cup,
\bar{b_1},\ldots,\bar{b_k} \models \bigcup_{i \in I} B_ib_{i_j}$,
\end{iteMize}
So the two extended disjoint union structures agree on all
$\sigma^*$-atomic formulas. Now relying on the semantics of
Boolean connectives, it can be shown by induction on the
complexity of quantifier free sentences that they also agree
on all Boolean combinations of atomic $\sigma^*$-sentences.

Inductive step: the inductive hypothesis says that whenever Duplicator has a
winning strategy in
$EF_\mso^n((\frM_i,\bar{A_i},\bar{a_i}),(\frN_i,\bar{B_i},\bar{b_i}))$
for all $1 \leq i \leq k$, he also has one in
$$EF_\mso^n((\bigoplus_{1\leq i\leq k}^f
\frM_i,(\bar{A_1}, \ldots,
\bar{A_k})^\cup,\bar{a_1},\ldots,\bar{a_k}), (\bigoplus_{ 1 \leq
i \leq k}^f \frN_i,(\bar{B_1},\ldots,\bar{B_k})^\cup,
\bar{b_1},\ldots,\bar{b_k})).$$ We want to show that this also
holds when the length of the games is $n+1$. Suppose Duplicator
has a winning strategy in the game
$EF_\mso^{n+1}((\frM_i,\bar{A_i},\bar{a_i}),(\frN_i,\bar{B_i},\bar{b_i}))$
for all $1 \leq i\leq k$. We describe Duplicator's answer to
Spoiler's first move in the game $EF_\mso^{n+1}((\bigoplus_{1\leq i\leq
k}^f \frM_i, \bar{A_1}, \ldots,
\bar{A_k},\bar{a_1},\ldots,\bar{a_k}), (\bigoplus_{ 1 \leq
i \leq k}^f \frN_i,\bar{B_1}, \ldots, \bar{B_k},
\bar{b_1},\ldots,\bar{b_k}))$.\\ It will then follow by induction
hypothesis, that he has a winning strategy in the remaining
$n$-length game.
\begin{iteMize}{$\bullet$}
\item   Spoiler's first move is a point move. Suppose Spoiler
picks $a$ in $\bigoplus_{1\leq i\leq k}^f \frM_i$. Then $a$ belongs to
$dom(\frM_i)$ for some $1\leq i\leq k$. So Duplicator uses his
winning strategy in
$EF_\mso^{n+1}((\frM_i,\bar{A_i},\bar{a_i}),(\frN_i,\bar{B_i},\bar{b_i}))$
to pick $b\in dom(\frN_i)$,
 so that he still has a winning strategy in $EF_\mso^n((\frM_i,\bar{A_i},\bar{a_i},a),(\frN_i,\bar{B_i},\bar{b_i},b))$.
 By induction hypothesis he also has one
 in the remaining $n$-length \mso game between the following two structures:
 
 $$(\bigoplus_{1\leq i\leq k}^f
\frM_i,(\bar{A_1},\ldots,\bar{A_k})^\cup,
\bar{a_1},\ldots,\bar{a_k},a)$$
and
$$(\bigoplus_{ 1 \leq i \leq k}^f \frN_i,
(\bar{B_1},\ldots,\bar{B_k})^\cup,
\bar{b_1},\ldots,\bar{b_k},b)$$
\item   Spoiler's first move is a set move. Suppose Spoiler chooses a set
$A$ in the set of admissible subsets of $\bigoplus_{1\leq
i\leq k}^f \frM_i$. Then $A$ is necessarily of the form $A_1 \cup
\ldots \cup A_k$, with $A_i$ an admissible subset of $\frM_i$. We
now define locally his response $B=B_1 \cup \ldots \cup B_k$,
using his winning strategies in the small structures, so that he still
has a winning strategy in
$EF_\mso^n((\frM_i,\bar{A_i},A_i,\bar{a_i}),(\frN_i,\bar{B_i},B_i,\bar{b_i}))$
for all $1 \leq i\leq k$. By induction hypothesis, he also has one
in the remaining $n$-length \mso game between the following two structures:

$$(\bigoplus_{1\leq i\leq k}^f
\frM_i,(\bar{A_1},A_1,\ldots,\bar{A_k},A_k)^\cup,
\bar{a_1},\ldots,\bar{a_k})$$
and
$$(\bigoplus_{ 1 \leq i \leq k}^f \frN_i,
(\bar{B_1},B_1,\ldots,\bar{B_k}, B_k)^\cup,
\bar{b_1},\ldots,\bar{b_k}).$$ 
(Note that this is enough,
because $A \in (\bar{A_1},A_1,\ldots,\bar{A_k},A_k)^\cup$.)\qed
\end{iteMize}\smallskip

\noindent Now an analogue of this result for disjoint unions can easily be derived as a corollary of Theorem \ref{msofusion}. For the convenience of the reader, we provide here the detailed argument:

\begin{cor}
Whenever $(\frM_i,\bar{A_i},
\bar{a_i})\equiv_\mso^n(\frN_i,\bar{B_i},\bar{b_i})$ for all $1
\leq i \leq k$ (with $\bar{a_i}$ a sequence of first-order
parameters of the form $a_{i_1},\ldots,a_{i_m}$ with
$m\in\mathbb{N}$ and $\bar{A_i}$ a sequence of set parameters of
the form $A_{i_1},\ldots,A_{i_{m'}}$ with $m'\in\mathbb{N}$,
similarly for the $\bar{b_i}$ and $\bar{B_i}$), then also
$\biguplus_{ 1 \leq i \leq k} \frM_i,(\bar{A_1},
\ldots, \bar{A_k})^\cup, \bar{a_1},\ldots,\bar{a_k}\equiv_\mso^n
\biguplus_{ 1 \leq i \leq k} \frN_i, (\bar{B_1},
\ldots, \bar{B_k})^\cup, \bar{b_1},\ldots,\bar{b_k}$.\label{msounion}
\end{cor}

\begin{proof}
Let $(\frM_i,\bar{A_i},\bar{a_i})\equiv_\mso^n(\frN_i,\bar{B_i},\bar{b_i})$ for all $1
\leq i \leq k$ (with $\bar{a_i}$ a sequence of first-order
parameters of the form $a_{i_1},\ldots,a_{i_m}$ with
$m\in\mathbb{N}$ and $\bar{A_i}$ a sequence of set parameters of
the form $A_{i_1},\ldots,A_{i_{m'}}$ with $m'\in\mathbb{N}$,
similarly for the $\bar{b_i}$ and $\bar{B_i}$).

Now consider the following expansions $\frM'_i$ and $\frN'_i$ of the $\sigma$ structures $\frM_i$ and $\frN_i$ to $\sigma^*=\sigma \cup \{Q_1,\ldots,Q_k\}$:
the interpretation of $Q_j$ is empty in $\frM'_i$ (respectively $\frN'_i$) whenever $i \neq j$ and it is the domain of $\frM'_i$ (respectively $\frN'_i$) whenever $i=j$.

Clearly $(\frM'_i,\bar{A_i},\bar{a_i})\equiv_\mso^n(\frN'_i,\bar{B_i},\bar{b_i})$ for all $1
\leq i \leq k$.

Now consider a mapping $f$ such that for every $n$-ary predicate $P \in \sigma^*$, $f(P)=Px_1 \ldots x_n$. By Theorem \ref{msofusion} we have that $$\bigoplus_{1\leq i\leq k}^f \frM'_i,(\bar{A_1},
\ldots, \bar{A_k})^\cup, \bar{a_1},\ldots,\bar{a_k}\equiv_\mso^n
\bigoplus_{ 1 \leq i \leq k}^f \frN'_i, (\bar{B_1},
\ldots, \bar{B_k})^\cup, \bar{b_1},\ldots,\bar{b_k}.$$

Corollary \ref{msounion} follows, because $$\bigoplus_{1\leq i\leq k}^f \frM'_i,(\bar{A_1},
\ldots, \bar{A_k})^\cup, \bar{a_1},\ldots,\bar{a_k}\text{ and }\bigoplus_{ 1 \leq i \leq k}^f \frN'_i, (\bar{B_1},
\ldots, \bar{B_k})^\cup, \bar{b_1},\ldots,\bar{b_k}$$ are isomorphic (w.r.t.~$\sigma$) to $$\biguplus_{ 1 \leq i \leq k} \frM_i,(\bar{A_1},
\ldots, \bar{A_k})^\cup, \bar{a_1},\ldots,\bar{a_k}\text{ and }\biguplus_{ 1 \leq i \leq k} \frN_i, (\bar{B_1},
\ldots, \bar{B_k})^\cup, \bar{b_1},\ldots,\bar{b_k}$$ respectively.
\end{proof}

Another important corollary of Theorem \ref{msofusion} is the fact that fusions of \mso-Henkin structures are also \mso-Henkin structures.
Let us stress the importance of this fact, which is needed for the correctness of our main completeness argument.

\begin{cor}
$\mathbb{A}_{\bigoplus_{1\leq i\leq k}^f\frM_i}$ is closed
under \mso parametric definability and so $\bigoplus_{1\leq i\leq k}^f\frM_i$ is a \mso-Henkin
structure.\label{msohenkinfusion}
\end{cor}

\begin{proof}
First note that the following are equivalent:
\begin{iteMize}{$\bullet$}
\item   $B$ is \mso parametrically definable in $\frM$,
\item   for some $n$, there is a finite sequence of parameters $\bar{a},\bar{A}$ such that
$B$ is defined by a \mso formula $\varphi$ of quantifier depth $n$
using $\bar{a},\bar{A}$,
\item  for some $n$, for every two points $a,a'\in
dom(\frM)$, if they are \mso $n$-indistinguishable using
$\bar{a},\bar{A}$, then $a \in B$ iff $a' \in B$.
\end{iteMize}
Now suppose for the sake of contradiction that there is $B \subseteq dom(\bigoplus_{1\leq
i\leq k}^f\frM_i)$ \mso parametrically definable in
$\bigoplus_{1\leq i\leq k}^f\frM_i$ using $\bar{a'},\bar{A'}$, but $B
\notin \mathbb{A}_{\bigoplus_{1\leq i\leq k}^f\frM_i}$. So
it means that for some $1\leq i\leq k $, $A_i=B \cap dom(\frM_i)$
is not \mso parametrically definable in $\frM_i$ i.e. there are
two \mso parametrically indistinguishable points $a \in B$, $a'
\notin B$. So for all $n$, for all sequence of parameters
$\bar{a},\bar{A}$ in $\frM_i$, $$(\frM_i,\bar{a},\bar{A},
a)\equiv_\mso^n(\frM_i,\bar{a},\bar{A}, a')$$ and by the fusion theorem,\footnote{There is no need to consider the case where
$\bar{a'},\bar{A'}$ is empty, because if a set is parametrically
definable using no parameter, it is also definable using
parameters.} $$\bigoplus_{1\leq i\leq k}^f\frM_i,
\bar{a},\bar{A},\bar{a'},\bar{A'},a\equiv^n_\mso \bigoplus_{1\leq i\leq k}^f\frM_i,
\bar{a},\bar{A},\bar{a'},\bar{A'},a'$$ But this entails that $B$
is not \mso parametrically definable in
$\bigoplus_{1\leq i\leq k}^f\frM_i$ using $\bar{a'},\bar{A'}$, which is a contradiction.
\end{proof}

\begin{cor}
$\mathbb{A}_{\biguplus_{ 1 \leq i \leq k} \frM_i}$ is closed
under \mso parametric definability and so $\biguplus_{ 1 \leq i \leq k} \frM_i$ is a \mso-Henkin
structure.\label{msohenkinunion}
\end{cor}
\begin{proof}
Analogous to the proof of Corollary \ref{msohenkinfusion} (as
$\mathbb{A}_{\bigoplus_{1\leq i\leq k}^f\frM_i}=\mathbb{A}_{\biguplus_{ 1 \leq i \leq k} \frM_i}$).
\end{proof}

Let us now consider the \fotc case. As $TC$ moves can only be played when there are already two pebbles on the board, it is more convenient to show first a version of our \fotc fusion theorem in which each small structure comes with at least two parameters. This allows us to define Duplicator's answer to a $TC$ move played in a big structure, by means of his winning strategies in the corresponding small structures. We then derive as a corollary the fusion theorem for non-parametrized structures.

\begin{thm}[Fusion Theorem for \fotc]
Let $\mathfrak{M}_i$ and $\mathfrak{N}_i$ be $\fotc$-Henkin
structures, where $1\leq i\leq k$. Furthermore, for $1\leq i\leq k$, 
let $\bar{a_i}$, $\bar{b_i}$ be sequences of first-order
parameters of the form $a_{i_1},\ldots,a_{i_m}$,
$b_{i_1},\ldots,b_{i_m}$ (where
$m\in\mathbb{N}$ may depend on $i$), where each sequence $\bar{a_i}$
(or  $\bar{b_i}$) contains at least two
distinct elements, unless the structure $\mathfrak{M}_i$
(respectively, $\mathfrak{N}_i$) has
only one element. Whenever $$(\frM_i,\bar{a_i})\equiv_\fotc^n(\frN_i,\bar{b_i})\text{ for
all }1 \leq i \leq k,$$
then also $$\bigoplus_{1\leq i\leq k}^f \frM_i,
\bar{a_1},\ldots,\bar{a_k}\equiv_\fotc^n \bigoplus_{1\leq
i\leq k}^f \frN_i,\bar{b_1},\ldots,\bar{b_k}.$$ 
\end{thm}

\proof
We define a winning strategy for Duplicator in the game
$$EF_\fotc^n((\bigoplus_{ 1 \leq i \leq k}^f \frM_i,
\bar{a_1},\ldots,\bar{a_k}), (\bigoplus_{1\leq i \leq k}^f
\frN_i, \bar{b_1},\ldots,\bar{b_k}))$$ out of her winning
strategies in the games
$EF_\fotc^n((\frM_i,\bar{a_i}),(\frN_i,\bar{b_i}))$ by induction
on $n$.

Base step: $n=0$, doing nothing is a strategy for Duplicator. We
need to show that the $\bigoplus_{ 1 \leq i \leq k}^f
\frM_i, \bar{a_1},\ldots,\bar{a_k}$ and $\bigoplus_{ 1
\leq i \leq k}^f \frN_i, \bar{b_1},\ldots,\bar{b_k}$ agree on all
atomic formulas. Now in the fusion structures, each atomic formula is
defined by $f$ in terms of a $\sigma^*$-quantifier free formula
that is evaluated in the corresponding disjoint union structure. So it
is enough to show that the disjoint union structures agree on all
atomic $\sigma^*$-formulas and on their Boolean combinations. The
initial match between the distinguished objects in $(\frM_i,
\bar{a_i})$ and $(\frN_i, \bar{b_i})$ is a partial isomorphism for
every $1\leq i \leq k$, so it is also one for $\biguplus_{1
\leq i \leq k} \frM_i, \bar{a_1},\ldots,\bar{a_k}$ and
$\biguplus_{1\leq i\leq k} \frN_i,
\bar{b_1},\ldots,\bar{b_k}$ i.e. the two disjoint union structures
agree on all $\sigma^*$-atomic formulas. Now relying on the
semantics of Boolean connectives, it can be shown by induction on
the complexity of quantifier free sentences that they also agree
on all Boolean combinations of atomic $\sigma^*$-sentences.

Inductive step: the inductive hypothesis says that whenever Duplicator
has a winning strategy in the game $EF_\fotc^n((\frM_i,\bar{a_i}),(\frN_i,\bar{b_i}))$
for some $(\frM_i,\bar{a_i}),(\frN_i,\bar{b_i})$ satisfying the required
conditions on parameters and $1 \leq i\leq k$, he also has one in
the game $EF_\fotc^n((\bigoplus_{ 1 \leq i \leq k}^f \frM_i,
\bar{a_1},\ldots,\bar{a_k}), (\bigoplus_{ 1 \leq i \leq
k}^f \frN_i,
\bar{b_1},\ldots,\bar{b_k}))$.

We want to show that this also holds whenever the length of the game
is $n+1$. Suppose Duplicator has a winning strategy in the game
$EF_\fotc^{n+1}((\frM_i,\bar{a_i}),(\frN_i,\bar{b_i}))$ for all $1
\leq i\leq k$. We describe Duplicator's answer to Spoiler's first
move in the game $EF_\fotc^{n+1}((\bigoplus_{ 1 \leq i \leq k}^f \frM_i,
\bar{a_1},\ldots,\bar{a_k}), (\bigoplus_{ 1 \leq i \leq
k}^f \frN_i, \bar{b_1},\ldots,\bar{b_k}))$. It will then follow by
induction hypothesis, that he has a winning strategy in the
remaining $n$-length game.
\begin{iteMize}{$\bullet$}
\item   Spoiler's first move is an $\exists$ move. Let Spoiler choose a point $a \in
dom(\bigoplus_{1 \leq i \leq k}^f \frM_i)$, then $a \in
dom(\frM_i)$ for some $1\leq i\leq k$. So Duplicator can use his
winning strategy in
$EF_\fotc^n((\frM_i,\bar{a_i}),(\frN_i,\bar{b_i}))$ and pick a
corresponding point $b$ in the other structure. Now he still has a
winning strategy in
$EF_\fotc^n((\frM_i,\bar{a_i},a),(\frN_i,\bar{b_i},b))$. So by
induction hypothesis he also has one in the remaining $n$ length
game $$EF_\fotc^n((\bigoplus_{1\leq i\leq k}^f \frM_i,
\bar{a_1},\ldots,\bar{a_k},a), (\bigoplus_{ 1 \leq i \leq
k}^f \frN_i, \bar{b_1},\ldots,\bar{b_k},b)).$$
\item   Spoiler's first move is a $TC$ move. Suppose Spoiler chooses a set
$A$ in the set of admissible subsets of $\bigoplus_{1 \leq
i \leq k}^f \frM_i$. Then $A$ is necessarily of the form $A_1 \cup
\ldots \cup A_k$, with $A_i$ an admissible subset (possibly empty)
of $\frM_i$. Her response $B=B_1 \cup \ldots \cup B_k$ can now be
defined locally for each $B_i$ using her winning strategies in the
small structures. So let Spoiler choose $A=A_1 \cup \ldots \cup A_k$.
Keeping in mind that each non single point small structure comes with
at least two distinct parameters, there are four cases:
\begin{enumerate}[a)]
\item     in $dom(\frM_i)$, there is a distinguished object inside, but also
outside $A_i$, so Duplicator considers $A_i$ together with these
two parameters and constructs $B_i$ by using his winning strategy
in $EF_\fotc^{n+1}((\frM_i,\bar{a_i}),(\frN_i,\bar{b_i}))$.
\item   in $dom(\frM_i)$, only distinguished objects exist
inside $A_i$\footnote{Note that as a special case we may have
$A_i=dom(\frM_i)$.}, so Duplicator considers any one of these
distinguished objects, say $a_j$, and looks at $A_i\backslash
\{a_j\}$ together with some parameter inside $A_i$.  Then he can
use his winning strategy in
$EF_\fotc^{n+1}((\frM_i,\bar{a_i}),(\frN_i,\bar{b_i}))$ to construct
an answer that we call $B_i'$. Now $B_i=B_i'\cup \{b_j\}$;
\item in $dom(\frM_i)$, only distinguished objects exist outside
  $A_i$,\footnote{Note that as a special case we may have
  $A_i=\emptyset$.} so Duplicator similarly considers some
  distinguished object $a_j$ and looks at $A_i \cup\{a_j\}$ together
  with some other parameter outside $A_i$, so that he can construct an
  answer that we call $B_i'$ by using his winning strategy in
  $EF_\fotc^{n+1}((\frM_i,\bar{a_i}),(\frN_i,\bar{b_i}))$. Now
  $B_i=B_i'\backslash \{b_j\}$;
\item    $\frM_i$ is a single point structure, then
$B_i=\emptyset$ if $A_i=\emptyset$ and $B_i=dom(\frM_i)$ if
$A_i=dom(\frN_i)$.
\end{enumerate}
Once $B=B_1 \cup \ldots \cup B_k$ has been constructed, Spoiler
picks two points $b \in B$ and $b' \notin B$. There are two cases:
\begin{enumerate}[1.]
\item  $b$ and $b'$ belong to the domain of one and the same small structure $\frN_i$ ; now $dom(\frM_i)$
is as previously described in $a), b), c)$ (but not $d))$, because
two distinct points cannot belong to one and the same single point
structure) and in each case Duplicator does the following:
\begin{enumerate}[a)]
\item He uses his winning strategy in the game
  $EF_\fotc^{n+1}((\frM_i,\bar{a_i}),(\frN_i,\bar{b_i}))$ to answer
  with $a, a'$, so that he still has a winning strategy in
  $EF_\fotc^n((\frM_i,\bar{a_i},a,a'),(\frN_i,\bar{b_i},b,b'))$. By
  induction hypothesis he also has one in the remaining $n$ length
  game $$EF_\fotc^n((\bigoplus_{1\leq i\leq k}^f \frM_i,
  \bar{a_1},\ldots,\bar{a_k},a,a'), (\bigoplus_{ 1 \leq i \leq k}^f
  \frN_i, \bar{b_1},\ldots,\bar{b_k},b,b')).$$
\item   Suppose initially that $b' \neq b_j$. Now Duplicator
considers $A_i \backslash \{a_j\}$ together with $a_j$ and
some other parameter inside this set.  Then he uses his winning strategy
in $EF_\fotc^{n+1}((\frM_i,\bar{a_i}),(\frN_i,\bar{b_i}))$ to pick
corresponding $a,a'$ in $\frM_i$, so that he still has a winning
strategy in
$EF_\fotc^n((\frM_i,\bar{a_i},a,a'),(\frN_i,\bar{b_i},b,b'))$. By
induction hypothesis he also has one in the remaining $n$ length
game $$EF_\fotc^n((\bigoplus_{1\leq i\leq k}^f \frM_i,
\bar{a_1},\ldots,\bar{a_k},a,a'), (\bigoplus_{ 1 \leq i
\leq k}^f \frN_i, \bar{b_1},\ldots,\bar{b_k},b,b'));$$ 
Next, suppose
$b=b_j$. Then we choose $a=a_j$. The parameter $a_j$ already matches
$b$ i.e., Duplicator has a winning strategy in
$$EF_\fotc^{n+1}((\frM_i,\bar{a_i},a),(\frN_i,\bar{b_i},b))$$ that he may
use to pick
$a'$, thus answering as if it was a point move (i.e., $a'$ has to be
$n$-equivalent to $b'$).  Therefore Duplicator still has a winning strategy
in $EF_\fotc^n((\frM_i,\bar{a_i},a,a'),(\frN_i,\bar{b_i},b,b'))$.
By induction hypothesis he also has one in the remaining $n$
length game $$EF_\fotc^n((\bigoplus_{1\leq i\leq k}^f \frM_i,
\bar{a_1},\ldots,\bar{a_k},a,a'), (\bigoplus_{1\leq i \leq
k}^f \frN_i, \bar{b_1},\ldots,\bar{b_k},b,b')).$$ This works, except that there is the additional condition
$a'\notin A_i$ that Duplicator must
also maintain in order to respect the rules of the game. A slightly more refined argument shows, however
that there
has to be an $n$-equivalent point to $b'$ which is outside $A_i$.
Indeed, instead of $b$, Spoiler could have picked any other point
$b^*\in B_i$ together with $b' \notin B_i$ and Duplicator's
winning strategy would have provided a correct answer $a^*\in
A_i$, $a'\notin A_i$, which means that Duplicator would have found
some point $a'$ which is at least $n$-equivalent to $b'$ and lies
outside $A_i$ (because if Duplicator has a winning strategy in
$EF_\fotc^n((\frM_i,\bar{a_i},a^*,a'),(\frN_i,\bar{b_i},b^*,b'))$
then he has one in
$EF_\fotc^n((\frM_i,\bar{a_i},a'),(\frN_i,\bar{b_i},b'))$ as well,
and consequently also in
$EF_\fotc^n((\frM_i,\bar{a_i},a,a'),(\frN_i,\bar{b_i},b,b'))$).
\item  Suppose initially that $b \neq b_j$.  Then Duplicator
considers $A_i \cup \{a_j\}$ together with $a_j$ and with some
other parameter outside this set and uses his winning strategy in
$EF_\fotc^{n+1}((\frM_i,\bar{a_i}),(\frN_i,\bar{b_i}))$, so that he
still has a winning strategy in
$EF_\fotc^n((\frM_i,\bar{a_i},a,a'),(\frN_i,\bar{b_i},b,b'))$. By
induction hypothesis he also has one in the remaining $n$ length
game $$EF_\fotc^n((\bigoplus_{1\leq i\leq k}^f \frM_i,
\bar{a_1},\ldots,\bar{a_k},a,a'), (\bigoplus_{1\leq i \leq
k}^f \frN_i, \bar{b_1},\ldots,\bar{b_k},b,b'));$$ otherwise
$b'=b_j$, then $a'=a_j$ because the parameter $a_j$ already
matches $b'$ i.e., Duplicator has a winning strategy in
$$EF_\fotc^{n+1}((\frM_i,\bar{a_i},a'),(\frN_i,\bar{b_i},b')),$$ so we
can show by a similar argument as above
that he can use it to pick $a\in A_i$, so that he still has a winning strategy in
$EF_\fotc^n((\frM_i,\bar{a_i},a,a'),(\frN_i,\bar{b_i},b,b'))$. By
induction hypothesis he also has one in the remaining $n$ length
game $$EF_\fotc^n((\bigoplus_{1\leq i\leq k}^f \frM_i,
\bar{a_1},\ldots,\bar{a_k},a,a'), (\bigoplus_{ 1 \leq i
\leq k}^f \frN_i, \bar{b_1},\ldots,\bar{b_k},b,b')).$$
\end{enumerate}
\item      otherwise $b \in dom(\frN_i,\bar{b_i})$ and $b' \in dom(\frN_j,\bar{b_j})$ with $i \neq j$;
we can again use a similar argument to show that Duplicator can
use his winning strategy in
$$EF_\fotc^{n+1}((\frM_i,\bar{a_i}),(\frN_i,\bar{b_i}))\text{ and }
EF_\fotc^{n+1}((\frM_j,\bar{a_j}),(\frN_j,\bar{b_j}))$$ to pick $a$,
$a'$ in the right part of the structure (that is, inside or outside
$A_i$), so that he still has a winning strategy in the games\\
$EF_\fotc^n((\frM_i,\bar{a_i},a),(\frN_i,\bar{b_i},b))$ and
$EF_\fotc^n((\frM_j,\bar{a_j},a'),(\frN_j,\bar{b_j},b'))$ (in the
special case where for instance, $\frM_j$ is a single point structure,
Duplicator picks the only available point in the other structure).
By induction hypothesis he also has one in the remaining $n$
length game $$EF_\fotc^n((\bigoplus_{1\leq i\leq k}^f \frM_i,
\bar{a_1},\ldots,\bar{a_k},a,a'), (\bigoplus_{ 1 \leq i
\leq k}^f \frN_i, \bar{b_1},\ldots,\bar{b_k},b,b')).\eqno{\qEd}$$
\end{enumerate}
\end{iteMize}

\noindent We now show a corollary of the preceding lemma, in which the
small structures do not come with any distinguished objects:
\begin{cor}
Whenever $\frM_i\equiv_\fotc^n\frN_i$ for all $1 \leq i \leq k$,
then also $\bigoplus_{1\leq i\leq k}^f \frM_i\equiv_\fotc^n
\bigoplus_{1\leq i\leq k}^f \frN_i$.\label{fotcfusion}
\end{cor}

\begin{proof}
We know that Spoiler's first two moves in
the \fotc-game of length $n+1$ between $\bigoplus_{1\leq i\leq k}^f \frM_i$ and
$\bigoplus_{1\leq i\leq k}^f \frN_i$ must be quantifier
moves, because the $TC$ move can only be played once there are two
pebbles on the board. Let us look at the first move. Suppose
Spoiler plays a point $a\in dom(\bigoplus_{1\leq i\leq k}^f
\frM_i)$. So $a \in dom(\frM_i)$ for some $1 \leq i \leq k$. By
Duplicator's winning strategy in $EF_\fotc^n(\frM_i, \frN_i)$, he
has an answer $b \in dom(\frN_i)$ such that $(\frM_i,a)
\equiv_\fotc^n (\frN_i,b)$. Let us rename $a$ with $a_{i_1}$ and
$b$ with $b_{i_1}$. Similarly, for every $j\neq i$ such that $1
\leq j \leq k$, fix some random point $a_{j_1}$ coming from the
domain of $\frM_j$, Spoiler could have played this point and so
Duplicator would have had an adequate answer $b_{j_1}$ such that
$(\frM_j,a_{j_1}) \equiv_\fotc^n (\frN_j,b_{j_1})$. Now for the
second round in the game, some point $a'=a_{l_2}$ or $b'=b_{l_2}$
coming from the domain of respectively $\frM_l$ or $\frN_l$ will
be played by Spoiler and Duplicator will be able to answer so that
$(\frM_l,a_{l_1},a_{l_2}) \equiv_\fotc^{n-2}
(\frN_l,b_{l_1},b_{l_2})$. Similarly, for each $\frM_j$ such that
$j \neq l$, we can find points such that $(\frM_j,a_{j_1},a_{j_2})
\equiv_\fotc^{n-2} (\frN_i,b_{j_1},b_{j_2})$. Now as for all $1
\leq i \leq k$, Duplicator has a winning strategy in
$EF_\fotc^{n-2}((\frM_i,a_{i_1},a_{i_2}),(\frN_i,b_{i_1},b_{i_2}))$,
by the previous lemma, he has one in
$$EF_\fotc^{n-2}(\bigoplus_{1\leq i\leq k}^f \frM_i, a_{1_1},
a_{1_2}, \ldots, a_{k_1}, a_{k_2}),(\bigoplus_{1\leq i\leq
k}^f \frN_i, b_{1_1}, b_{1_2}, \ldots, b_{k_1}, b_{k_2})),$$ so
he also has one in $EF_\fotc^{n-2}(\bigoplus_{1\leq i\leq k}^f
\frM_i, a, a'),(\bigoplus_{1\leq i\leq k}^f \frN_i, b,
b'))$.
\end{proof}

\begin{cor}
Whenever $\frM_i\equiv_\fotc^n\frN_i$ for all $1 \leq i \leq k$,
then it also holds that $\biguplus_{ 1 \leq i \leq k} \frM_i\equiv_\fotc^n
\biguplus_{ 1 \leq i \leq k} \frN_i$.\label{fotcunion}
\end{cor}

\begin{proof}
Analogous to the proof of Corollary \ref{msounion}.
\end{proof}

\begin{cor}
$\mathbb{A}_{\bigoplus_{1\leq i\leq k}^f \frM_i}$ is closed
under \fotc parametric definability and so the structure $\bigoplus_{1\leq i\leq
k}^f\frM_i$ is a \fotc-Henkin
structure.\label{fotchenkinfusion}
\end{cor}

\begin{proof}
Analogous to the proof of Corollary \ref{msohenkinfusion}.
\end{proof}

\begin{cor}
$\mathbb{A}_{\biguplus_{ 1 \leq i \leq k} \frM_i}$ is closed
under \fotc parametric definability and so the structure $\biguplus_{ 1 \leq i \leq k} \frM_i$ is a \fotc-Henkin
structure.\label{fotchenkinunion}
\end{cor}
\begin{proof}
Analogous to the proof of Corollary \ref{msohenkinunion}.
\end{proof}

In the \folfp case, the situation parallels the \fotc case. As $LFP$ moves can only be played when there is already one pebble on the board, it is more convenient to show first a version of our \folfp fusion theorem in which each small structure comes with at least one \fo parameter. This allows us to define Duplicator's answer to a $LFP$ move played in the big structure, by means of his winning strategies in the small structures. We then derive as a corollary the fusion theorem for non-parametrized structures.

\begin{thm}[Fusion Theorem for \folfp]
Let $\bar{a_i}$, $\bar{b_i}$ be \emph{non empty} sequences of first-order
parameters of the form $a_{i_1},\ldots,a_{i_m}$, $b_{i_1},\ldots,b_{i_m}$, with
$m\in\mathbb{N}$ and $\bar{A_i}$, $\bar{B_i}$ sequences of set parameters of
the form $A_{i_1},\ldots,A_{i_{m'}}$, $B_{i_1},\ldots,B_{i_{m'}}$ with $m'\in\mathbb{N}$.
Whenever
$$(\frM_i,\bar{A_i},\bar{a_i})\equiv_\folfp^n(\frN_i,\bar{B_i},\bar{b_i})\text{ for all }1\leq i\leq k,$$
then also $$\bigoplus_{1\leq i\leq k}^f \frM_i,
(\bar{A_1},\ldots,\bar{A_k})^\cup,\bar{a_1},\ldots,\bar{a_k}\equiv_\folfp^n
\bigoplus_{1\leq i\leq k}^f
\frN_i,(\bar{B_1},\ldots,\bar{B_k})^\cup,\bar{b_1},\ldots,\bar{b_k}.$$
\end{thm}

\proof
We proceed by induction on $n$, defining a winning strategy for Duplicator in the game $EF_\folfp^n((\bigoplus_{1\leq i\leq k}^f \frM_i,
(\bar{A_1},\ldots,\bar{A_k})^\cup,\bar{a_1},\ldots,\bar{a_k}),(\bigoplus_{1\leq i \leq k}^f
\frN_i,(\bar{B_1},\ldots,\bar{B_k})^\cup,
\bar{b_1},\ldots,\bar{b_k}))$,\\
 out of her winning strategies in
the games
$EF_\folfp^n((\frM_i,\bar{A_i},\bar{a_i}),(\frN_i,\bar{B_i},\bar{b_i}))$.

Base step: $n=0$, doing nothing is a strategy for Duplicator (this can be justified by a similar argument as in the \mso case).

Inductive step: the inductive hypothesis says that whenever Duplicator has a winning strategy
in
$EF_\folfp^n((\frM_i,\bar{A_i},\bar{a_i}),(\frN_i,\bar{B_i},\bar{b_i}))$
for pairs of structures
$(\frM_i,\bar{A_i},\bar{a_i}),(\frN_i,\bar{B_i},\bar{b_i})$
satisfying the required conditions on parameters with $1 \leq i\leq
k$, he also has one in $EF_\folfp^n((\bigoplus_{1\leq i\leq
k}^f
\frM_i,(\bar{A_1},\ldots,\bar{A_k})^\cup,\bar{a_1},\ldots,\bar{a_k}),
(\bigoplus_{ 1 \leq i \leq k}^f \frN_i,
(\bar{B_1},\ldots,\bar{B_k})^\cup,\bar{b_1},\ldots,\bar{b_k}))$.

We want to show that this also holds when the length of the games
is $n+1$. Suppose Duplicator has a winning strategy in
$EF_\folfp^{n+1}((\frM_i,\bar{A_i},\bar{a_i}),(\frN_i,\bar{B_i},\bar{b_i}))$
for all $1 \leq i\leq k$. We describe Duplicator's answer to
Spoiler's first move in the \folfp-game of length $n+1$ in between $(\bigoplus_{1\leq i\leq
k}^f \frM_i,(\bar{A_1},\ldots,\bar{A_k})^\cup,
\bar{a_1},\ldots,\bar{a_k})$ and $(\bigoplus_{1\leq i\leq k}^f
\frN_i,(\bar{B_1},\ldots,\bar{B_k})^\cup,
\bar{b_1},\ldots,\bar{b_k})$. It then follows by induction
hypothesis, that he has a winning strategy in the remaining
$n$-length game.
\begin{iteMize}{$\bullet$}
\item   Spoiler's first move is an $\exists$ move.

Same argument as for \mso and \fotc.
\item   Spoiler's first move is a $\forall$ move.

Symmetric.
\item   Spoiler's first move is a GFP move.

Suppose Spoiler chooses a set $A$ in the set of admissible subsets
of $\bigoplus_{1 \leq i\leq k}^f \frM_i$ with some pebble
$a_{i_j} \in A$. Then $A$ is necessarily of the form $A_1 \cup
\ldots \cup A_k$, with $A_i$ an admissible subset of $\frM_i$. Her
response $B=B_1 \cup \ldots \cup B_k$ can now be defined locally
for each $B_i$ using her winning strategies in the small structures.
So let Spoiler choose $A=A_1 \cup \ldots \cup A_k$. Keeping in
mind that each small structure comes with at least one parameter,
there are four cases:
\begin{enumerate}[1)]
\item     in $dom(\frM_i)$, there is a distinguished object inside $A_i$ and $A_i \neq dom(\frM_i)$,
so Duplicator considers $A_i$ together with this parameter and
constructs $B_i$ by using his winning strategy in
$EF_\folfp^{n+1}((\frM_i,\bar{A_i},\bar{a_i}),(\frN_i,\bar{B_i},\bar{b_i}))$.
\item  in $dom(\frM_i)$, there are only distinguished objects
outside $A_i \neq \emptyset$, so Duplicator considers
any $a_j$ among those and looks at
$A_i\cup \{a_j\}$, so that he can use his winning strategy in
$EF_\folfp^{n+1}((\frM_i,\bar{A_i},\bar{a_i}),(\frN_i,\bar{B_i},\bar{b_i}))$
to construct an answer that we call $B_i'$. Now
$B_i=B_i'\backslash \{b_j\}$. This is a correct answer, because
the (posimorphism) condition to be maintained (see Definition \ref{effolfp}) is that for every
pebble $a_l$ on the board at the end of the game, $a_l \in A_i
\Rightarrow b_l \in B_i$. But by Duplicator's winning strategy in
$EF_\folfp^{n+1}((\frM_i,\bar{A_i},A_i\cup
\{a_j\},\bar{a_i}),(\frN_i,\bar{B_i},B_i',\bar{b_i}))$, we
know already that for every such pebble, $a_l \in A_i\cup \{a_j\}
\Rightarrow b_l \in B_i'$, so also $a_l \in A_i \Rightarrow b_l
\in B_i' \backslash \{b_j\}$, since the winning conditions will
  assure that $a_l = a_j$ if and only if $b_l = b_j$.
\item  $B_i=dom(\frM_i)$. So
$A_i=dom(\frN_i)$. As pebbles are only chosen using Duplicator's
winning strategies in the small structures, the posimorphism condition
will be maintained.
\item     $B_i=\emptyset$. So $A_i=\emptyset$. As no pebble can belong to this set, the
posimorphism condition will be maintained.
\end{enumerate}
Now that $B=B_1 \cup \ldots \cup B_k$ has been constructed,
Spoiler picks a new element $b \in B$ which belongs to the domain
of one particular small structure $\frN_i$ (so $b \in B_i$) and
$dom(\frM_i)$ is as previously described either in $1$, $2$ or
$3$ (but not $4)$, because $b$ cannot belong to the empty set)
and in each case Duplicator does the following:
\begin{enumerate}[1)] 
\item    Duplicator answers with $a$ according to his
winning strategy in
$$EF_\folfp^{n+1}((\frM_i,\bar{A_i},\bar{a_i}),(\frN_i,\bar{B_i},\bar{b_i}));$$
\item  Duplicator again considers $A_i\cup \{a_j\}$ and
answers according to his winning strategy in
$EF_\folfp^{n+1}((\frM_i,\bar{A_i},A_i\cup
\{a_j\},\bar{a_i}),(\frN_i,\bar{B_i},B_i',\bar{b_i}))$. This
is safe, because the pebble to be chosen may be assumed to be fresh, so it
won't be $a_j$;
\item   Duplicator picks
a random pebble $a_j\in dom(\frM_i)$ and considers
$dom(\frM_i)\backslash \{a_j\}$. His winning strategy provides him
with a correct answer.
\end{enumerate}
So in any of these cases (either $1$, $2$ or $3$), Duplicator has a
winning strategy in the game
$EF_\folfp^n((\frM_i,\bar{A_i},A_i,\bar{a_i},a),(\frN_i,\bar{B_i},B_i,\bar{b_i},b))$.
Now for all $j\neq i$, $1 \leq j \leq k$, he also has one in
$EF_\folfp^n((\frM_j,\bar{A_j},A_j,\bar{a_j}),(\frN_j,\bar{B_j},B_j,\bar{b_j}))$.
So by induction hypothesis, he has one in the remaining $n$-length \folfp
game between the following two structures:

$$(\bigoplus_{1\leq i\leq k}^f \frM_i,
(\bar{A_1},A_1,\ldots,\bar{A_k},A_k)^\cup,\bar{a_1},\ldots,\bar{a_k},a)$$
and
$$(\bigoplus_{1\leq i\leq k}^f
\frN_i,(\bar{B_1},B_1,\ldots,\bar{B_k},B_k)^\cup,
\bar{b_1},\ldots,\bar{b_k},b)$$
\item   Spoiler's first move is a $LFP$ move.

Symmetric.\qed
\end{iteMize}

\begin{cor}
Whenever $\frM_i\equiv_\folfp^n\frN_i$ for all $1 \leq i \leq k$,
then it also holds that $\bigoplus_{1\leq i\leq k}^f
\frM_i\equiv_\folfp^n \bigoplus_{1\leq i\leq k}^f
\frN_i$.\label{folfpfusion}
\end{cor}

\begin{proof}
We know that Spoiler's first move in
$EF_\folfp^{n+1}(\bigoplus_{1\leq i\leq k}^f \frM_i,
\bigoplus_{1\leq i\leq k}^f \frN_i)$ must be a \fo
quantifier move, because the $LFP$ move can only be played once
there is a pebble on the board. Let us look at the first move.
Suppose Spoiler plays a point $a\in dom(\bigoplus_{1\leq
i\leq k}^f \frM_i)$. So $a \in dom(\frM_i)$ for some $1 \leq i
\leq k$. By Duplicator's winning strategy in $EF_\folfp^n(\frM_i,
\frN_i)$, he has an answer $b \in dom(\frN_i)$ such that
$(\frM_i,a) \equiv_\folfp^n (\frN_i,b)$. Let us rename $a$ with
$a_i$ and $b$ with $b_i$. Similarly, for every $j\neq i$ such that
$1 \leq j \leq k$, fix some random point $a_j$ coming from the
domain of $\frM_j$, Spoiler could have played this point and so
Duplicator would have had an adequate answer $b_j$ such that
$(\frM_j,a_j) \equiv_\folfp^n (\frN_j,b_j)$. Now as for all $1 \leq
i \leq k$, Duplicator has a winning strategy in
$EF_\folfp^{n-1}((\frM_i,a_i),(\frN_i,b_i))$, by the previous lemma,
he has one in $EF_\folfp^{n-1}(\bigoplus_{1\leq i\leq k}^f
\frM_i, a_1, \ldots, a_k),(\bigoplus_{1\leq i\leq k}^f
\frN_i, b_1, \ldots, b_k))$, so he also has one in
$EF_\folfp^{n-1}(\bigoplus_{1\leq i\leq k}^f \frM_i,
a),(\bigoplus_{1\leq i\leq k}^f \frN_i, b))$.
\end{proof}

\begin{cor}
Whenever $\frM_i\equiv_\folfp^n\frN_i$ for all $1 \leq i \leq k$,
then it also holds that $\biguplus_{ 1 \leq i \leq k}\frM_i\equiv_\folfp^n
\biguplus_{ 1 \leq i \leq k}\frN_i$.\label{folfpunion}
\end{cor}

\begin{proof}
Analogous to the proof of Corollary \ref{msounion}.
\end{proof}

\begin{cor}
$\mathbb{A}_{\bigoplus_{1\leq i\leq k}^f}$ is closed
under \folfp parametric definability and so the structure $\bigoplus_{1\leq i\leq k}^f\frM_i$ is a \folfp-Henkin
structure.\label{folfphenkinfusion}
\end{cor}

\begin{proof}
Analogous to the proof of Corollary \ref{msohenkinfusion}.
\end{proof}

\begin{cor}
$\mathbb{A}_{\biguplus_{ 1 \leq i \leq k} \frM_i}$ is closed
under \folfp parametric definability and so the structure $\biguplus_{ 1 \leq i \leq k} \frM_i$ is a \folfp-Henkin
structure.
\label{folfphenkinunion}
\end{cor}
\begin{proof}
Analogous to the proof of Corollary \ref{msohenkinunion}.
\end{proof}

\section{Putting it Together: Completeness on Finite Trees}
\label{sec5}

\subsection{Forests and Operations on Forests}

In Section \ref{sectproof}, we will prove that no $\Lambda$-sentence can
distinguish $\Lambda$-Henkin-models of $\vdash_\Lambda^{tree}$
from standard models of $\vdash_\Lambda^{tree}$. More precisely,
we will show that for each $n$, every definably
well-founded $\Lambda$-quasi-tree is $n$-$\Lambda$-equivalent to a finite
tree. In order to give an inductive proof, it will be more convenient
to consider a stronger version of this result concerning a class of finite and infinite Henkin
structures that we call \textit{quasi-forests}. In this section, we give the definition of quasi-forest and we show how they can be combined into bigger quasi-forests using the notion of fusion from Section \ref{sec4}. Whenever quasi forests are finite, we simply call them \emph{finite
forests}. As a simple example, consider a finite tree and remove the root node, then it is no longer a finite tree. Instead it is a finite sequence of trees, whose roots stand in a linear (sibling)
order.\footnote{Note that, as far as roots are concerned, two
nodes can be siblings without sharing any parent. This would not
happen in a quasi tree.} It does not have a unique root, but it does have a
unique \emph{left-most root}. For technical reasons it will be convenient in the definition of quasi forests to add an extra monadic predicate $R$ labeling the roots.

\begin{defi}[$\Lambda$-quasi-forest]
Let $T=(dom(T),<,\prec,P_1,\ldots,\ldots P_n, \mathbb{A}_T)$ be a
$\Lambda$-quasi-tree. Given a node $a$ in $T$, consider the $\Lambda$-substructure of
$T$ generated by the set $\{x ~|~\exists z (a \preceq z\wedge z \leq x)\}$, which is the set formed by $a$
together with all its siblings to the right and their descendants.
The $\Lambda$-quasi-forest $T_a$ is obtained by labeling each
root in this substructure with $R$ ($Rx \Leftrightarrow_{def} \neg
\exists y~ y<x$). Whenever $T$ is a tree, we simply call $T_a$ a
forest.
\end{defi}

We will show in our main proof of completeness that for each $n$ and for each node $a$ in a
definably well-founded $\Lambda$-quasi-tree, the $\Lambda$-quasi-forest $T_a$ is $n$-$\Lambda$-equivalent to a finite forest.
Our proof will use a notion of composition of $\Lambda$-quasi-forests which is a special case of fusion.
Given a single node forest $F_1$ and two $\Lambda$-quasi-forests $F_2$ and
$F_3$, we construct a new $\Lambda$-quasi-forest $\bigoplus^{f^\triangle}(F_1,F_2,F_3)$ by letting
the unique element in $F_1$ be the left-most root, the roots of $F_2$ become the children of this node and the roots of $F_3$ become its siblings to the right. We then derive a corollary of the $\Lambda$-fusion theorem for compositions of $\Lambda$-quasi-forests and use it in Section \ref{sectproof}.

\begin{defi}
Let $\sigma = \{<, \prec, R, P_1, \ldots, P_n\}$, be a relational vocabulary with only monadic predicates except $<$ and $\prec$. Given three
additional monadic predicates $Q_1, Q_2, Q_3$, we define a mapping $f^\triangle$ from $\sigma$ to quantifier-free formulas over $\sigma \cup \{Q_1,
Q_2, Q_3\}$ by letting
\begin{iteMize}{$\bullet$}
\item   $f^\triangle(P_i)=P_i(x_1)$
 \item  $f^\triangle(<) = x_1<x_2 \vee (Q_1(x_1) \wedge Q_2(x_2))$
 \item  $f^\triangle(\prec)= x_1 \prec x_2 \vee (Q_1(x_1) \wedge Q_3(x_2) \wedge R(x_2))$
\item   $f^\triangle(R)=(Q_3(x_1) \wedge R(x_1)) \vee Q_1(x_1)$
\end{iteMize}
\end{defi}
\begin{cor}
Let $F_1$ be a single node forest and $F_2$, $F_3$ $\Lambda$-quasi forests. If $F_2\equiv_\Lambda^n F'_2$ and $F_3\equiv_\Lambda^n F'_3$ then
$\bigoplus^{f^\triangle}(F_1,F_2,F_3)\equiv_\Lambda^n \bigoplus^{f^\triangle}(F_1,F'_2,F'_3)$.
\end{cor}

\begin{figure}[!h]
\begin{center}
\begin{tikzpicture}[->,>=stealth',scale=0.8, thin, level/.style={sibling distance=30mm/#1}]
   \node (1) at (-3.5,3) {$F_1$};
   \node (1) at (-6.5,0.5) {$F_2$};
   \node (1) at (9.5,2) {$F_3$};

 \draw(-2,2.5) ellipse (1cm and 0.5cm) ;
 
 \draw(-1.5,0.3) ellipse (4cm and 1.7cm) ;
 
 \draw(6.1,0.7) ellipse (3.5cm and 2.3cm) ;

 \draw[*->] (-2,2.5) -- (-2,1);
 \draw[->] (-2,2.5) -- (1,1);
 \draw[->] (-2,2.5) -- (-5,1);

 \draw[*->] (-2,1) -- (-2,-0.5);
 \draw[*->] (1,1) -- (1,-0.5);
 \draw[->] (1,1) -- (2,-0.5);
 \draw[->] (1,1) -- (0,-0.5);

 \draw[*->] (-5,1) -- (-5,-0.5);

\draw[*->] (4,1) -- (4,-0.5);
\draw[*->] (4,2.5) -- (4,1);

 \draw[*->] (6,2.5) -- (6,1);

 \draw[*->] (8,2.5) -- (8,1);
 \draw[->] (8,1) -- (7,-0.5);
 \draw[*->] (8,1) -- (8,-0.5);
 \draw[->] (8,1) -- (9,-0.5);
   
   \end{tikzpicture}
\caption{A composition of forests using the mapping $f^\triangle$}\label{compforest} \label{jolidessin}
  \end{center}
\end{figure}

Figure \ref{jolidessin} represents a composition of three forests $F_1$, $F_2$, $F_3$ which uses the mapping $f^\triangle$. Only new $<_{ch}$-arrows are represented, linking the unique node in $F_1$ to the root nodes in $F_2$. But new $\preceq$-links have also been added and the roots in $F_3$ have became the siblings to the right of the root in $F_1$. This is implicitly indicated by the left to right organization of the picture.

\subsection{Main Proof of Completeness}
\label{sectproof}

\begin{lem}
  For all $n\in\mathbb{N}$, every definably well-founded $\Lambda$-quasi-tree of
  finite signature is $n$-$\Lambda$-equivalent to a finite tree. In particular,
  a $\Lambda$-sentence is valid on definably well-founded $\Lambda$-quasi-trees iff it is
  valid on finite trees.
\label{lemma3}
\end{lem}

\begin{proof}
  Let $T$ be a $\Lambda$-quasi-tree, without loss of generality assume that a monadic predicate $R$ labels its root (and only that node in the tree). During this proof, it will be convenient to work with
  $\Lambda$-quasi-forests. Note that finite $\Lambda$-quasi-forests
  are simply finite forests and finite $\Lambda$-quasi-trees are
  simply finite trees (cf.~Proposition~\ref{prop:define-finite} for
  the case of quasi-trees, from which the case for quasi-forests
  follows immediately). Let $X_n$ be the set of all nodes $a$ of $T$ for
  which it holds that $T_a$ is $n$-$\Lambda$-equivalent to a finite
  forest. We first show that `belonging to $X_n$'' is a property definable in $T$ (Claim $1$).
  We then use the induction scheme to show that every node of a definably well-founded $\Lambda$-quasi-tree (so in particular, the root) has this property (Claim $2$).
  \medskip\par\noindent\emph{Claim 1:}
    $X_n$ is invariant for $n+1$-$\Lambda$-equivalence (i.e.,
    $(T,a)\equiv_{n+1}^\Lambda (T,b)$ implies that $a\in X_n$ iff $b\in X_n$), and
    hence is defined by a $\Lambda$-formula of
    quantifier depth $n+1$.

  \begin{proofofclaim}
    Suppose that $(T,a)\equiv_{n+1}^\Lambda(T,b)$.
    We will show that $T_a\equiv_{n}^\Lambda T_b$, and hence, by the
    definition of $X_n$, $a\in X_n$ iff $b\in X_n$.
     By the definition of $\Lambda$-quasi-forests, $dom(T_a)=\{x ~|~\exists z (a \preceq z\wedge z \leq x)\}$. Let $\varphi$ be any $\Lambda$-sentence of quantifier depth $n$. We can assume without loss of generality that $\varphi$ does not contain the variables $z$ and $x$ (otherwise we can rename in $\varphi$ these two variables). By lemma \ref{rel}, $(T,a)\models REL(\varphi,\exists z (a \preceq z\wedge z \leq x),x)$ iff $T_a \models \varphi$.
    Notice that $REL(\varphi,\exists z (a \preceq z\wedge z \leq x),x)$ expresses precisely that $\varphi$
    holds in $(T,a)$ within the subforest $T_a$. Moreover, the quantifier depth of
    $REL(\varphi,\exists z (a \preceq z\wedge z \leq x)$ is at most $n+1$. It follows that $(T,a)\models REL(\varphi,\exists z (a \preceq z\wedge z \leq x),x)$
    iff $(T,b)\models REL(\varphi,\exists z (b \preceq z\wedge z \leq x),x)$, and hence $T_a\models\varphi$ iff
    $T_b\models\varphi$.

    For the second part of the claim, note that by Lemma \ref{finitenesslemma}, up to logical equivalence, there are only
    finitely many $\Lambda$-formulas of any given quantifier depth,
    as the vocabulary is finite.
  \end{proofofclaim}

  \smallskip\par\noindent\emph{Claim 2:} If all descendants and siblings
  to the right of $a$ belong to $X_n$, then $a$ itself
  belongs to $X_n$.

  \begin{proofofclaim}
    Let us consider the case where $a$ has both a
    descendant and a following sibling (all other cases are
    simpler). Then, by axioms T3, T5, T8, T9 and T10, $a$ has a first child $b$, and
    an immediate next sibling $c$. Moreover, we know that both $b$ and
    $c$ are in $X_n$. In other words, $T_b$ and $T_c$ are
    $n$-$\Lambda$-equivalent to finite forests $T'_b$ and $T'_c$. Now, we
    construct a finite $\Lambda$-quasi-forest $T'_a$ by taking a $f^\triangle$-fusion of $T'_b$, $T'_c$ and of the $\Lambda$-substructure of $T$ generated by $\{a\}$, whose unique element becomes a
    common parent of all roots of $T'_b$ and a left sibling of all
    roots of $T'_c$. So we get $T'_a = \bigoplus^{f^\triangle}(T \upharpoonright \{a\},T'_b,T'_c)$).
    It is not hard to see that $T'_a$ is again a finite
    forest. Moreover, by the fusion theorem,
    $\bigoplus^{f^\triangle}(T \upharpoonright \{a\},T_b,T_c)) \equiv_n^\Lambda
    T'_a$.
    Now to show that $\bigoplus^{f^\triangle}(T \upharpoonright
    \{a\},T_b,T_c))$ is isomorphic to $T_a$ (which entails $T_a \equiv_n^\Lambda T'_a$ i.e. $T_a$ is $n$-$\Lambda$-equivalent to a finite forest),
    it is enough to show $\mathbb{A}_{T_a}=\mathbb{A}_{\bigoplus^{f^\triangle}(T \upharpoonright \{a\},
    T_b,T_c)}$. It holds that $\mathbb{A}_{\bigoplus^{f^\triangle}(T \upharpoonright \{a\},
    T_b,T_c)} \subseteq \mathbb{A}_{T_a}$ because we can define in $T_a$ each such union of sets by means
    of a disjunction. Now to show $\mathbb{A}_{T_a} \subseteq \mathbb{A}_{\bigoplus^{f^\triangle}(T\upharpoonright \{a\},
    T_b,T_c)}$, take $A \in \mathbb{A}_{T_a}$, so $A=X \cap T_a$ for some $X \subseteq \mathbb{A}_T$. As $dom(T\upharpoonright \{a\})$,
    $dom(T_b)$ and $dom(T_c)$ are all definable in $T$, the intersection of each of these sets with $A$ is also definable in $T$ and hence 
    $A \cap dom(T\upharpoonright \{a\})$ is definable in $dom(T\upharpoonright \{a\})$, $A \cap dom(T_b)$ is definable in $T_b$ and 
    $A \cap dom(T_c)$ is definable in $T_c$. But then $A$, which can be formed as the union of these three sets, is also definable in $\mathbb{A}_{\bigoplus^{f^\triangle}(T\upharpoonright \{a\},
    T_b,T_c)}$.
    
    \end{proofofclaim}

  It follows from these two claims, by the induction scheme for
  definable properties, that $X_n$ contains all nodes of the $\Lambda$-quasi-tree,
  including the root, and hence $T$ is $n$-$\Lambda$-equivalent to a finite
  tree (to a finite forest actually, but the root of the $\Lambda$-quasi-tree being labeled by $R$, it can also be viewed as a  $\Lambda$-quasi-forest).
  For the second statement of the lemma, it suffices to
  note that every $\Lambda$-sentence has a finite vocabulary and a finite
  quantifier depth.
\end{proof}

\begin{thm}\label{thm:L-completeness}
  Let $\Lambda \in \{\mso,\fotc,\folfp\}$. The $\Lambda$-theory of finite trees is completely axiomatized by $\vdash_\Lambda^{tree}$.
\end{thm}

\begin{proof}
Theorem~\ref{thm:L-completeness} follows directly from Lemma \ref{lemma3} and Corollary \ref{henktreecomp}.
\end{proof}

\subsection{Definability of the Class of Finite Trees}
\label{sec53}

Proposition \ref{tcdeftree} below shows together with Theorem~\ref{thm:L-completeness} that on standard structures, the set of $\vdash_\Lambda^{tree}$ consequences actually defines the (not \fo-definable) class of finite trees. That is, $\vdash_\Lambda^{tree}$ has \emph{no infinite standard model} at all.

\begin{prop}[\cite{2006}]
Let $\Lambda \in \{\fotc,\folfp,\mso\}$. On standard structures, there is a $\Lambda$-formula which defines the class of finite trees.
\label{tcdeftree}
\end{prop}
\proof[Sketch of the proof]
It is enough to show it for $\Lambda=\fotc$. It follows by Section \ref{exppow} that it also holds for \mso and \folfp.
We merely give a sketch of the proof. For additional details we
refer the reader to \cite{2006}. 

Recall from Proposition~\ref{prop:define-finite} that, on finite standard
structures, the finite conjunction of the axioms T1--T10 in
Figure~\ref{fig:tree-axioms} defines the class of
finite trees, i.e., any finite structure satisfying this
conjunction is a finite tree. Now we will explain how to construct
another sentence, which together with this one, actually defines
on arbitrary standard structures the class of finite trees. Let
$L$ be a shorthand for the formula labeling the leaves in the
tree ($Lx \Leftrightarrow_{def} \neg \exists y~ x<y$) and $R$ a
shorthand for the formula labeling the root ($Rx
\Leftrightarrow_{def} \neg \exists y~  y<x$). Consider the
depth-first left-to-right ordering of nodes in a tree and the
\fotc formula $\varphi(x,y)$ saying ``the node that comes after
$x$ in this ordering is $y$'':
\begin{center}
$\varphi(x,y):\approx (\neg Lx \wedge x <_{ch}y \wedge \neg
\exists z z \prec y) \vee(Lx \wedge x \prec_{ns}y) \vee(Lx \wedge
\neg \exists z x \prec z \wedge \exists z(z<x \wedge z \prec_{ns}y
\wedge \neg \exists w w<x \wedge z<w \wedge \exists u w \prec_{ns}
u))$
\end{center}
There is also a \fotc formula which says that ``x is the very last
node in this ordering''. $\varphi(x,y)$ can be combined with this
formula into an $\fotc$ formula $\chi$ expressing that the tree is
finite by saying that (we rely here for the interpretation of
$\chi$ on the alternative semantics for the $TC$ operator given in
Proposition \ref{tcaltsem}) ``there is a finite sequence of nodes
$x_1 \ldots x_n$ such that $x_1$ is the root, $x_{i+1}$ the node
that comes after $x_i$ in the above ordering, for all $i$, and
$x_n$ is the very last node of the tree in the above ordering''.
\[
\chi:\approx \exists u \exists z (Rz \wedge
[TC_{xy}\varphi](z,u)\wedge \neg \exists u' (u \neq u' \wedge
[TC_{xy} \varphi](u,u')))\eqno{\qEd}
\]

\begin{thm}\label{thm:define}
For $\Lambda\in\{\mso,\fotc,\folfp\}$, the set of axioms T1--T10 together
with all $\Lambda$-instances of the induction scheme Ind defines the
class of finite trees.
\end{thm}

\begin{proof}
By Proposition \ref{tcdeftree} we can express in $\Lambda$ by means of some formula $\chi$ that a structure is a finite tree. So $\chi$
is provable in $\vdash_\Lambda^{tree}$ (as it is a $\Lambda$-formula
valid on the class of finite trees). In other words, if $\Gamma$ is
the set of axioms T1--T10 together
with all $\Lambda$-instances of the induction scheme Ind, then we have 
that $\Gamma\vdash_\Lambda \chi$.
\end{proof}

\section{Finite Linear Orders}
\label{secisftseclo}

Let us note that a simplified version of this method can be used in order to show the completeness of \mso, \fotc and \folfp on finite node-labeled linear orders (i.e., finite node-labeled trees in which every node has at most one child). The relevant simpler axioms are the ones listed in Figures \ref{fig:FO-axioms}, \ref{fig:lo-axioms} and respectively, Figures \ref{fig:MSO-axioms}, \ref{fig:fotc-axioms} and \ref{fig:folfp-axioms}.

\begin{figure}[!h]
\hrule\smallskip
\begin{tabular}{@{}ll@{\hspace{14mm}}l@{}}
L1. & $\forall xyz(x<y \land y<z \to x<z)$ & $<$ is transitive \\
L2. & $\neg \exists x (x< x)$ & $<$ is irreflexive \\
L3. &$\forall xy(x<y \to \exists z (x<_{ch}z \land z\leq y))$ & immediate children\\
L4. &$\exists x \forall y \neg(y < x)$ & there is a root \\
L5. &$\forall xy(x=y \lor x<y \lor y<x)$ & $<$ is total \\
Ind. &  $\forall x(\forall y((x <y \to \varphi(y)) \to \varphi(x)) \to \forall x\varphi(x)$\\
 & \\
where &  \\
 & $\varphi(x)$ ranges over $\Lambda$-formulas in one free variable $x$\\
and & \\
  & \multicolumn{2}{l}{$x<_{ch}y$ is shorthand for $x<y \land \neg \exists z(z<y\wedge x<z)$} \\
 \end{tabular}
\smallskip\hrule
\caption{Specific axioms on finite linear orders} 
\label{fig:lo-axioms}
\end{figure}

\section{Conclusion}
\label{conclmso}

In this paper, taking inspiration from Kees Doets \cite{1987} we developed a uniform method for obtaining complete axiomatizations of fragments of \mso on finite trees. For that purpose, we had to adapt classical tools and notions from finite model theory to the specificities of Henkin semantics. The presence of admissible subsets called for some refinements in model theoretic constructions such as formation of substructure or disjoint union. Also, we noticed that not every Ehrenfeucht-Fra\"{\i}ss\'e game that has been used for \fotc was suitable to use on Henkin-structures. We focused on a game which does not seem to have been used previously in the literature. We also established analogues of the \fo Feferman-Vaught theorem for \mso, \fotc and \folfp on Henkin-structures (let us recall that related work for the case of standard structures can be found in \cite{2004}). We considered fusions, a particular case of the Feferman-Vaught notion of generalized product and obtained results for Henkin-structures which might be interesting to generalize and use in other contexts.

We applied our method to \mso, \fotc and \folfp, but it would be worth also examining other fragments of \mso or logics such as monadic deterministic transitive closure logic (\fodtc, which was advocated in \cite{1219706} as particularly relevant in the context of applications to model-theoretic syntax) or monadic alternating transitive closure logic (\foatc), see also \cite{1992}.

An important feature of our main completeness argument (the idea of which was borrowed from Kees Doets) is the way we used the inductive scheme of Figure \ref{fig:tree-axioms}. Hence, extending our approach to another class of finite structures would involve finding a comparable scheme. We also know that we should focus on a logic which is decidable on this class, as on finite structures recursive enumerability is equivalent to decidability (as long as the model-checking is decidable). This suggests that other natural candidates would be fragments of \mso on classes of finite structures with bounded treewidth.

Finally, let us notice that \mso is also known to be decidable over infinite trees and over linear orders of order type $\omega$. It would be interesting to look for a model-theoretic argument which would work on a Henkin model and produce an intended model of one of these theories in a way comparable to what we did here or to what Keisler did in \cite{keisler}.
Note that related complete axiomatizations of monadic theories of classes of infinite structures can be found in \cite{siefkesbuchi},  \cite{siefkes} and \cite{zaiontz}, but that instead of relying on Henkin-semantics, the completeness proofs there are based on automata-theoretic techniques.


\begin{thebibliography}{10}

\bibitem{DBLP:books/aw/AbiteboulHV95}
Serge Abiteboul, Richard Hull, and Victor Vianu.
\newblock {\em Foundations of Databases}.
\newblock Addison-Wesley, 1995.

\bibitem{1995}
Rolf Backofen, James Rogers, and Krishnamurti Vijay-Shankar.
\newblock {A} {F}irst-{O}rder {A}xiomatization of the {T}heory of {F}inite
  {T}rees.
\newblock {\em Journal of Logic, Language and Information}, 4(4):5--39, 1995.

\bibitem{modal}
Patrick Blackburn, Maarten de~Rijke, and Yde Venema.
\newblock {\em Modal Logic}.
\newblock Cambridge University Press, Cambridge, 2000.

\bibitem{736408}
Uwe Bosse.
\newblock An ``{E}hrenfeucht-{F}ra\"{i}ss\'{e} game'' for fixpoint logic and
  stratified fixpoint logic.
\newblock In {\em CSL '92: Selected Papers from the Workshop on Computer
  Science Logic}, pages 100--114, London, UK, 1993. Springer-Verlag.

\bibitem{siefkesbuchi}
J.~Richard B\"{u}chi and Dirk Siefkes.
\newblock {\em Decidable Theories: Vol. 2: The Monadic Second Order Theory of
  All Countable Ordinals}.
\newblock Lectures Notes in Mathematics. Springer, Berlin, Heidelberg, 1973.

\bibitem{1992}
A.~Calo and Johann~A. Makowsky.
\newblock The {E}hrenfeucht-{F}ra\"{\i}ss\'e games for transitive closure.
\newblock {\em Lecture Notes in Computer Science}, 620:57--68, 1992.

\bibitem{TenCate2006}
{Balder ten} Cate.
\newblock The expressivity of xpath with transitive closure.
\newblock In {\em PODS}, pages 328--337, 2006.

\bibitem{DBLP:conf/fossacs/CateF10}
{Balder ten} Cate and Ga{\"e}lle Fontaine.
\newblock {A}n {E}asy {C}ompleteness {P}roof for the {M}odal $\mu$-{C}alculus
  on {F}inite {T}rees.
\newblock In {\em FOSSACS}, pages 161--175, 2010.

\bibitem{07}
{Balder ten} Cate and Maarten Marx.
\newblock {A}xiomatizing the {L}ogical {C}ore of {X}{P}ath 2.0.
\newblock In {\em ICDT}, pages 134--148, 2007.

\bibitem{CateMarxXPath20}
{Balder ten} Cate and Maarten Marx.
\newblock {A}xiomatizing the {L}ogical {C}ore of {X}{P}ath 2.0.
\newblock {\em Theory of Computing Systems}, 44(4):561--589, 2009.

\bibitem{1376952}
{Balder ten} Cate and Luc Segoufin.
\newblock Transitive closure logic, nested tree walking automata, and
  {X}{P}ath.
\newblock {\em J. ACM}, 57(3), 2010.

\bibitem{1987}
Kees Doets.
\newblock {\em Completeness and Definability : Applications of the Ehrenfeucht
  Game in Second-Order and Intensional Logic}.
\newblock PhD thesis, Universiteit van Amsterdam, 1987.

\bibitem{ebfl95}
Heinz-Dieter Ebbinghaus and J\"{o}rg Flum.
\newblock {\em Finite Model Theory}.
\newblock Perspectives in Mathematical Logic. Springer-Verlag, Berlin, 1995.

\bibitem{Enderton}
Herbert~B. Enderton.
\newblock {\em A Mathematical Introduction to Logic}.
\newblock Harcourt - Academic Press, San Diego, NY, USA, 2001.
\newblock Second edition.

\bibitem{1959}
Solomon Feferman and Robert Vaught.
\newblock The first-order properties of algebraic systems.
\newblock {\em Fundamenta Mathematicae}, 47:57--103, 1959.

\bibitem{GottlobKoch}
Georg Gottlob and Christoph Koch.
\newblock Monadic datalog and the expressive power of languages for web
  information extraction.
\newblock In {\em Proceedings of {PODS} 2002}, pages 17--28, 2002.

\bibitem{736267}
Erich Gr\"{a}del.
\newblock {O}n {T}ransitive {C}losure {L}ogic.
\newblock In {\em CSL '91: Proceedings of the 5th Workshop on Computer Science
  Logic}, pages 149--163, London, UK, 1992. Springer-Verlag.

\bibitem{1950}
Leon Henkin.
\newblock {C}ompleteness in the {T}heory of {T}ypes.
\newblock {\em The Journal of Symbolic Logic}, 15(2):81--91, 1950.

\bibitem{97}
Wilfrid Hodges.
\newblock {\em A Shorter Model Theory}.
\newblock Cambridge University Press, New York, NY, USA, 1997.

\bibitem{keisler}
Jerome Keisler.
\newblock Logic with the quantifier ``there exists uncountably many''.
\newblock {\em An. Math. Logic}, 1:1--93, 1970.

\bibitem{1047003}
Stephan Kepser.
\newblock {Q}uerying {L}inguistic {T}reebanks with {M}onadic {S}econd-{O}rder
  {L}ogic in {L}inear {T}ime.
\newblock {\em J. of Logic, Lang. and Inf.}, 13(4):457--470, 2004.

\bibitem{2006}
Stephan Kepser.
\newblock Properties of {B}inary {T}ransitive {C}losure {L}ogic over {T}rees.
\newblock In Giorgio~Satta Paola~Monachesi, Gerald~Penn and Shuly Wintner,
  editors, {\em Formal Grammar 2006}, pages 77--89, 2006.

\bibitem{1998}
Gr\'{e}gory Lafitte and Jacques Mazoyer.
\newblock Th\'{e}orie des mod\`{e}les et complexit\'{e}.
\newblock Technical report, Ecole Normale Sup\'{e}rieure de Lyon, septembre
  1998.

\bibitem{leo}
Leonid Libkin.
\newblock {\em Elements of Finite Model Theory (Texts in Theoretical Computer
  Science. An Eatcs Series)}.
\newblock SpringerVerlag, 2004.

\bibitem{2004}
Johann~A. Makowsky.
\newblock {A}lgorithmic uses of the {F}eferman {V}aught {T}heorem.
\newblock {\em Annals of Pure and Applied Logic}, 126(1--3):159--213, 2004.

\bibitem{230876}
Mar\'{\i}a Manzano.
\newblock {\em Extensions of First Order Logic}.
\newblock Cambridge University Press, New York, NY, USA, 1996.

\bibitem{521965}
James Rogers.
\newblock {\em Descriptive Approach to Language-Theoretic Complexity}.
\newblock CSLI Publications, Stanford, CA, USA, 1998.

\bibitem{DBLP:conf/icalp/Schweikardt04}
Nicole Schweikardt.
\newblock On the {E}xpressive {P}ower of {M}onadic {L}east {F}ixed-{P}oint
  {L}ogic.
\newblock In {\em ICALP}, pages 1123--1135, 2004.

\bibitem{siefkes}
Dirk Siefkes.
\newblock {A}n {A}xiom {S}ystem for the {W}eak {M}onadic {S}econd {O}rder
  {T}heory of {T}wo {S}uccessors.
\newblock {\em Israel Journal of Mathematics}, 30(3):264--284, 1978.

\bibitem{1219706}
Hans-J\"{o}rg Tiede and Stephan Kepser.
\newblock {M}onadic {S}econd-{O}rder {L}ogic and {T}ransitive {C}losure
  {L}ogics over {T}rees.
\newblock {\em Electron. Notes Theor. Comput. Sci.}, 165:189--199, 2006.

\bibitem{zaiontz}
Charles Zaiontz.
\newblock {A}xiomatization of the {M}onadic {T}heory of {O}rdinals $<\omega^2$.
\newblock {\em Mathematical Logic Quarterly}, 29(6):337--356, 1983.



\end{thebibliography}
\end{document}